\documentclass[journal,onecolumn]{IEEEtran}

\usepackage{color}

\usepackage{algorithm}
\usepackage{algpseudocode}

\usepackage{url}
\usepackage[bookmarks=false]{hyperref} 
\usepackage[anythingbreaks]{breakurl}

\usepackage{cite}
\usepackage{graphicx}
\usepackage{amssymb,hhline,enumerate,dsfont}
\usepackage{amsmath}
\usepackage{array}
\usepackage{psfrag}
\usepackage{epsfig}
\usepackage{amsthm}
\usepackage{amsmath,amssymb,times}
\usepackage{graphicx,multirow,epsfig,cite}

\usepackage[utf8]{inputenc}
\usepackage{multirow}
\usepackage{mathtools}

\def\block(#1,#2)#3{\multicolumn{#2}{c}{\multirow{#1}{*}{$ #3 $}}}

\newtheorem{thm}{Theorem}
\newtheorem{rmrk}{Remark}
\newtheorem{lem}{Lemma}

\newtheorem{defn}{Definition}

\theoremstyle{definition}
\newtheorem{examplex}{Example}
\newenvironment{example}
  {\pushQED{\qed}\examplex}
  {\popQED\endexamplex}

\begin{document}
\title{Product Matrix MSR Codes with Bandwidth Adaptive Exact Repair}


\author{\IEEEauthorblockN{Kaveh Mahdaviani\IEEEauthorrefmark{1}, Soheil Mohajer\IEEEauthorrefmark{2}, and Ashish Khisti\IEEEauthorrefmark{1} }\\
\IEEEauthorblockA{\IEEEauthorrefmark{1}ECE Dept., University of Toronto, Toronto, ON M5S3G4, Canada\\
\IEEEauthorrefmark{2}ECE Dept., University of Minnesota, Minneapolis, MN 55455, USA \\              
Email: \{kaveh, akhisti\}@comm.utoronto.ca,~soheil@umn.edu}}

\maketitle

\begin{abstract}
In a distributed storage systems (DSS) with $k$ systematic nodes, robustness against node failure is commonly provided by storing redundancy in a number of other nodes and performing repair mechanism to reproduce the content of the failed nodes. Efficiency is then achieved by minimizing the storage overhead and the amount of data transmission required for data reconstruction and repair, provided by coding solutions such as regenerating codes \cite{Regenerating}. Common explicit regenerating code constructions enable efficient repair through accessing a predefined number, $d$, of arbitrary chosen available nodes, namely helpers. In practice, however, the state of the system dynamically changes based on the request load, the link traffic, \emph{etc.}, and the parameters which optimize system's performance vary accordingly. It is then desirable to have coding schemes which are able to operate optimally under a range of different parameters simultaneously. Specifically, adaptivity in the number of helper nodes for repair is of interest. While robustness requires capability of performing repair with small number of helpers, it is desirable to use as many helpers as available to reduce the transmission delay and total repair traffic.

In this work we focus on the minimum storage regenerating (MSR) codes, where each of the $n$ nodes in the network is supposed to store $\alpha$ information units, and the source data of size $k\alpha$ could be recovered from any arbitrary set of $k$ nodes. We introduce a class of MSR codes that realize optimal repair bandwidth simultaneously with a set of different choices for the number of helpers, namely $D=\{d_{1}, \cdots, d_{\delta}\}$. Our coding scheme follows the Product Matrix (PM) framework introduced in \cite{PM_Codes}, and could be considered as a generalization of the PM MSR code presented in \cite{PM_Codes}, such that any $d_{i} = (i+1)(k-1)$ helpers can perform an optimal repair. As a result, the coding rate in our construction is limited by $\frac{k}{n}\leq \frac{1}{2}$. However, similar to the original design of PM MSR codes, our solution can realize practical values of the parameter $\alpha$. Recently \cite{ExplicitMSR} has presented another explicit MSR coding scheme which is capable of performing optimal repair for various number of helpers. The solution presented in \cite{ExplicitMSR} works for any arbitrary set of parameters $k, D$ and can achieve high coding rates, but the required $\alpha$ for this code is exponentially large. We show that the required value for $\alpha$ in the coding scheme presented in this work is exponentially smaller when compared to \cite{ExplicitMSR} for the same set of other parameters. Particularly, for a DSS with $n$ nodes, and $k$ systematic nodes, the required value for $\alpha$ is reduced from $s^n$ to $sk$, where $s=\mathrm{lcm}\left(d_{1}-k+1,\cdots,d_{\delta}-k+1\right)$. We also show the required field size in the presented coding scheme is equal to $n$\footnote{This work was presented in part at the 2017 IEEE Information Theory Workshop (ITW), Kaohsiung, Taiwan \cite{BAMSR_ITW17}.}.

\end{abstract}

\IEEEpeerreviewmaketitle

\section{Introduction}

Distributed storage systems (DSS) are compelling solutions to the fast growth of the demand in storage and accessibility of digital content. The main feature of such systems is to provide accessibility and durability for the stored data by introducing redundancy. In parallel, the number of storage components as well as the number of users connecting to these servers are dramatically increasing. These increase the chance of failures in the system, due to storage device failures or inaccessible nodes with overloaded traffic. Replication is the simplest approach to make the distributed storage system robust against such failures, which is implemented in systems such as \cite{Hadoop, HDFS_Guide, HDFS_Guide}. This approach however, provides simplicity in the cost of huge storage overhead. 

In the literature of erasure codes, there exists solutions, such as the Reed-Solomon (RS) code \cite{RS_Codes}, which offer similar fault-tolerance level as the replication does, with significantly less storage overhead. For instance, a 3-replication based DSS needs to accommodate two mirrors for every single storage node in order to achieve reliability guarantee against any simultaneous two node failure, which arises $200\%$ storage overhead, while \emph{maximum distance separable} (MDS) codes, such as RS code, achieve the same level of fault-tolerance guarantee by adding only two redundant storage nodes for the whole system. In other words, an erasure code with \emph{MDS property} can guarantee to recover the whole source data from any subset of stored encoded segments of collective size equal to the source data size.

On the other hand, when a node fails in a DSS, it needs to be replaced by a new node in order to maintain the system's performance. Such procedure is referred to as a \emph{repair}. To perform a repair, the system downloads some data from a subset of the surviving nodes, called \emph{helper nodes}. The amount of data downloaded for a repair is referred to as the \emph{repair bandwidth}. In conventional MDS erasure codes, such as the RS codes, one has to reconstruct the entire source data and re-encode it to recover a single lost segment. Hence, the repair bandwidth of these codes is at least as large as the size of the whole source data.

Considering both storage overhead and repair bandwidth simultaneously, the family of \emph{regenerating codes} \cite{Regenerating_Conf, Regenerating} offers a very efficient coding mechanism for distributed storage systems. More precisely, a regenerating code on a Galois field $\mathbb{F}_{q}$ for a DSS with $n$ storage nodes, maps the source data of size $F$ symbols into $n$ pieces of size $\alpha$ symbols each, and stores them in $n$ separate nodes, such that any $k$ out of $n$ nodes suffice to recover the data. Such a system is capable of tolerating up to $(n-k)$ node failures. Moreover, upon failure of one node, it can be replaced by a new node whose content of size $\alpha$ symbols is determined by connecting to an arbitrary set of $d$ (where $d \geq k$) helper nodes, and downloading $\beta$ symbols form each (where $\beta \leq \alpha$). Hence, the total repair bandwidth in regenerating codes is $d\beta$, which is denoted by $\gamma$. The parameter $\alpha$ is also referred to as the \emph{subpacketization level}.

Ideally, one would like to minimize the storage overhead, and repair bandwidth simultaneously. It turns out that for a given file size $F$, there is a trade-off between the subpacketization level $\alpha$ and the repair bandwidth $\gamma$, and one can be minimized only at the cost of a penalty for the other \cite{Regenerating}. In  particular, at one extreme point of this trade-off, one could first minimize the subpacketization level, $\alpha$, and then minimize the \emph{per-node repair bandwidth}, $\beta$, to obtain a \emph{minimum storage regenerating} (MSR) code. As a result, MSR codes have the MDS property, and also minimize the repair bandwidth for the given $\alpha$ \cite{Regenerating}, which means for an MSR code we have $F=k\alpha$, and 
\begin{align}\label{eq_betaMSR}
\beta = \frac{F}{k(d-k+1)}.
\end{align}
The total repair bandwidth ($d\beta$) of an MSR code is then upper-bounded by that of RS code, and only coincides with that when $d=k$. In other word, the total repair bandwidth is decreasing super-linearly as $d$ grows in MSR codes. 

Reversing the order of minimization between $\alpha$, and $\beta$ results in another extreme point of the trade-off. Such codes are not MDS and hence have more storage overhead but provide the smallest possible repair bandwidth and are referred to as \emph{minimum repair bandwidth} (MBR) codes. Our focus in this work is on MSR codes as they minimize the storage cost.

In general, two distinct types of repair can be identified: in an \emph{exact repair} scenario the replacement generated node will contain the same data as stored in the failed node. In the so called \emph{functional repair}, however, the replacement node may store a different content, provided that the entire new system maintain the same properties as of the original one. In practice, exact repair regenerating codes are much more appealing, mainly due to the fact that they could provide systematic encoding, which is a significant advantage in practice.

It is shown that for functional repair MSR code design could be translated into a linear network coding problem \cite{Regenerating}. Design of such codes with exact repair property is more challenging, due to a large number of constraints that should be simultaneously satisfied. Various constructions for exact repair MSR codes have been proposed for various sets of parameters \cite{PM_Codes, Exact_AI_Asymptotic, Permutation, Zigzag_Codes, Polynimial, Alternative, Genral_sub_packet, Small_field, ExplicitMSR_Tamo, ExplicitMSR, ExplicitMSR_nearOptimal, Coupled_Layer}.

Among the numerous available regenerating code designs, the common adopted model considers a rigid predetermined number $d$ (where $k \leq d \leq n-1$) of helpers required for any repair procedure. Each of these helpers is also assumed to provide $\beta = \gamma/d$ repair symbols. This sets a single threshold for the system's capability to perform the repair. On the other hand, in practice the state of the system dynamically changes as a function of various factors including availability of nodes, traffic load, available bandwidth, \emph{etc.} It has been shown that when such characteristics change in the system the optimal number of helper nodes for minimizing the cost of repair changes \cite{Kermarrec_System_Perspective}. Therefore, runtime adaptation would be of great value towards optimizing the performance. For instance, when the system is heavily loaded by many read requests, many nodes in the network might not be capable to provide the required repair bandwidth and hence would be considered unavailable for the repair. This may result in system's failure to perform the repair, while there might be a few nodes, less than $d$, which are capable of providing even more than $\beta$ repair symbols. A natural question is whether one can download more data from the available helpers, and accomplish the repair without the struggler helpers. Besides, from (\ref{eq_betaMSR}) it is clear that in optimal repair of MSR codes, increasing the number of helpers reduces both $\beta$, and $\gamma=d\beta$. One could consider a situation in which there are many nodes, more than $d$, which are capable of contributing to the repair. It would be of great practical value then if the system could increase the number of helpers and reduce both the per-node and total repair traffic. Note that this will also reduce the transmission delay, which is one of the main bottlenecks in the DSS's performance. We refer to such property as \emph{bandwidth adaptivity}. Note that the dynamic capability of service for storage nodes is a well-known characteristic for many practical distributed systems such as peer-to-peer systems or heterogeneous network topologies \cite{SpaceMonkey, Tahoe, CFS, Google, OceanStore, TotalRecall}.

The design of such codes has been of interest and the significance of bandwidth adaptivity in the performance of the system has been emphasised in \cite{Flexible_Regenerating, BWAdapt_MFR, BWAdapt_Kermarrec ,Exact_AI_Asymptotic, BWAdapt_Opportunistic, ProgressiveEngagement}. However, these works mainly focus on the study of fundamental limits and appealing properties of regenerating codes with bandwidth adaptivity, such as the significant improvement in the mean-time-to-data-loss (MTTDL), and the coding schemes presented in these works either only provide functional repair \cite{Flexible_Regenerating, BWAdapt_MFR, BWAdapt_Kermarrec, BWAdapt_Opportunistic}, or relax some of the constraints and consequently yield to a sub-optimum trade-of \cite{Exact_AI_Asymptotic, ProgressiveEngagement}.

When considering the optimal exact repair, it is a challenging problem to design such coding scheme with a large flexibility degree since it needs to satisfy many optimality conditions simultaneously. As a result, this problem has only been considered for the MSR \cite{Exact_AI_Asymptotic, ExplicitMSR}, and MBR \cite{BAER_ISIT16} extreme points of the trade-off. For the MBR case, \cite{BAER_ISIT16} provides an exact repair bandwidth adaptive solution for a wide range of practical parameters based on the \emph{Product Matrix} framework introduced in \cite{PM_Codes}. In \cite{Exact_AI_Asymptotic} a bandwidth adaptive code is provided based on interference alignment, which only achieves the MSR characteristics when both $\alpha$ and $\beta$ tend to infinity. The first explicit exact repair MSR code constructions which satisfy the bandwidth adaptivity for finite values of $\alpha$ and $\beta$ are introduced in \cite{ExplicitMSR}. These constructions work for any parameters $k$, $n$, and all values of $d$ such that $k<d<n$. However, the required values for the parameters $\alpha$, and $\beta$ in these constructions are still very huge (\emph{i.e.} exponentially large in $n$), and hence they only achieve optimality for extremely large contents. Recently, \cite{ExplicitMSR_nearOptimal} introduced a modified version of the code constructions in \cite{ExplicitMSR} which achieves MSR optimality for much lower values of $\alpha$, at the cost of losing bandwidth adaptivity. Indeed the MSR code in \cite{ExplicitMSR_nearOptimal} works only for $d=n-1$. We will review the related works in more details in Section \ref{Sec_Related_Works}.

In this work we will address the design problem of MSR codes with bandwidth adaptive exact repair for small $\alpha$, and $\beta$, following the Product Matrix framework \cite{PM_Codes}. For any positive integer design parameter $\delta$, the presented coding scheme allows us to choose the number of helper nodes for each repair scenario from the set $\{2(k-1), \cdots, (\delta+1)(k-1)\}$, based on availability of the nodes, network traffic and load state of the nodes. For any choice of the number of helpers, the per-node repair bandwidth, $\beta$, takes its optimum value based on (\ref{eq_betaMSR}). To this end we change the design of the \emph{``message matrix"} for the PM MSR codes introduced in \cite{PM_Codes} to enable its partitioning into smaller symmetric submatices in many different ways. We then also design appropriate repair and reconstruction schemes using a successive cancellation mechanism, to enable performing repair based on different choices for the number of helpers, while keeping the required contribution of each helper at the optimal value. 

Product Matrix codes are practical regenerating codes due to their small subpacketization and repair bandwidth requirements. Our focus in this work is to generalize the PM MSR codes to enable bandwidth adaptive repair, while keeping small practical values for the parameters $\alpha$, and $\beta$. For the construction presented in this work the required values for $\alpha$ and $\beta$ are linear in terms of the number of nodes in the DSS, which results in an exponential reduction compared to the constructions proposed in \cite{ExplicitMSR}. However, unlike the constructions in \cite{ExplicitMSR} which achieve any coding rate $k/n$, it should be mentioned that the coding rate in our construction is limited to $k/n \leq 1/2$. The main contributions of this work are explained in the next section, after formally defining the problem setup.

The rest of this paper is organized as follows: The following section formally introduced the problem setup and summarizes the main contributions. Section \ref{Sec_Related_Works} briefly reviews the most relevant works in the literature. Coding scheme and examples are presented in Section \ref{Sec_Coding_Scheme}, which is followed by a discussion on the properties of the code in Section \ref{Sec_Discussion}. Conclusion and appendix are provided at the end.

\section{Model and Main Results}

\subsection{Model}
In this section we will briefly introduce a setup for the distributed storage system and the coding scheme of our interest. This model is a modified version of the original setup considered in \cite{Regenerating}, and is very similar to the model considered in \cite{BAER_ISIT16}.

The first element we consider for the model of our bandwidth adaptive distributed storage system is a predefined Galois filed alphabet, $\mathbb{F}_{q}$ of size $q$. Hereafter we assume all the symbols stored or transmitted through the network are elements of $\mathbb{F}_q$. Besides, we consider a homogeneous group of $n$ storage nodes, each capable of storing $\alpha$ symbols.

\begin{defn}[Bandwidth Adaptive Regenerating Code]
Consider parameters $\alpha$, $n$, $k$, $\delta$, a set $D=\{d_{1}, \cdots, d_{\delta}\}$, with $d_{1}<\cdots<d_{\delta}$, and a \emph{total repair bandwidth function} $\gamma:D\rightarrow [\alpha,\infty)$. A bandwidth adaptive regenerating code $\mathcal{C}(n$, $k$, $D$, $\alpha$, $\gamma)$ is a regenerating code with subpacketization level $\alpha$, such that following mechanisms are guaranteed.
\begin{itemize}
\item \emph{Repair:} In each repair process the number of helpers, $d$, can be chosen arbitrarily from the set $D$. The choice of helper nodes is also arbitrary, and each of the chosen helpers then provides $\beta(d) = \gamma(d) / d$ repair symbols.
\item \emph{Data Reconstruction:} The data collector recovers the whole source data by accessing any arbitrary set of $k$ nodes.
\end{itemize}
\end{defn}

Note that the flexibility of the repair procedure depends on the parameter $\delta$, such that for a larger $\delta$, there are more options to select the number of helpers. In general, it is appealing to have small choices such as $d_{1}$, to guarantee the capability of code to perform repair when the number of available helpers is small, and also large choices such as $d_{\delta}$, to provide the capability of reducing the total as well as per-node repair bandwidth, and hence the transmission delay, whenever a larger number of helpers are available. The coding scheme we present in this work allows to design such a range for the elements in $D$. 
%

\begin{defn}[Total Storage Capacity]
For a set of parameters $\alpha$, $n$, $k$, $\delta$, a set $D=\{d_{1}, \cdots, d_{\delta}\}$, and a given function $\gamma:D \rightarrow [\alpha,\infty)$, the \emph{total storage capacity} of a bandwidth adaptive regenerating code, $\mathcal{C}(n$, $k$, $D$, $\alpha$, $\gamma)$, is the maximum size of a file that could be stored in a network of $n$ storage nodes with subpacketization level $\alpha$, using a bandwidth adaptive regenerating code $\mathcal{C}(n$, $k$, $D$, $\alpha$, $\gamma)$. We will denote the storage capacity of such a system by $F(n$, $k$, $D$, $\alpha$, $\gamma)$, or simply $F$, when the parameters could be inferred from the context.
\end{defn}

\begin{defn}[Bandwidth Adaptive MSR Codes, and Flexibility Degree]
For any choice of parameters $\alpha$, $n$, $k$, $\delta$, and set $D=\{d_{1}, \cdots, d_{\delta}\}$, the bandwidth adaptive regenerating codes that realize both the MDS property defined by 
\begin{align}\label{eq_MDSProperty}
F(n, k, D, \alpha, \gamma) = k \alpha,
\end{align}
as well as the the MSR characteristic equation simultaneously for all $d \in D$, given as
\begin{align}\label{eq_BWAMSR}
\alpha = (d-k+1)\beta(d), ~~ \forall{d \in D},
\end{align}
are referred to as \emph{bandwidth adaptive MSR} codes. Moreover, the number of elements in the set $D$ is referred to as \emph{flexibility degree} of the code, and is denoted by $\delta$.
\end{defn}

\begin{rmrk}
It is worth mentioning that in a bandwidth adaptive MSR code, the total repair bandwidth,
\begin{align}
d\beta(d) = \frac{d\alpha}{d-k+1}, \nonumber
\end{align}
is a decreasing function of $d$. As a result, naive time-sharing between multiple MSR code components fails to satisfy optimality in a bandwidth adaptive setting.
\end{rmrk}

\subsection{Main Results}

The main contribution of this work is to provide a bandwidth adaptive MSR coding scheme with small subpacketization level, and field size requirement. This coding scheme also guarantees exact repair of any failed node with many different choices of the number of helpers. This result is formally stated in the following theorem. In this paper $\mathrm{lcm}()$ denotes the least common multiple.

\begin{thm}\label{Thm_main}
For arbitrary positive integers $n$, $k$, and $\delta$, there exists an adaptive bandwidth MSR code over a Galois field $\mathbb{F}_{q}$ of size $q \geq n$, with subpacketization level $\alpha$ and total storage capacity $F$, satisfying 
\begin{align}\label{eq_alpha_Thm}
\alpha = (k-1)\mathrm{lcm}\left(1,2,\cdots,\delta \right), ~~~~ F = k \alpha,
\end{align}
which is capable of performing exact repair using any $d_{i}$ helpers, for 
\begin{align}
d_{i}=(i+1)(k-1),~i \in \{1, \cdots, \delta\}, \nonumber
\end{align}
and simultaneously satisfies the MSR characteristic equation (\ref{eq_BWAMSR}) for all $d_{i}$. \emph{i.e.,}
\begin{align}
\beta(d_{i}) = \frac{\alpha}{(d_{i}-k+1)},~ i \in \{1, \cdots, \delta\}. \nonumber
\end{align}
\end{thm}

We provide a constructive proof for this theorem based on introducing an explicit construction for a bandwidth adaptive MSR coding scheme in Section \ref{Sec_Coding_Scheme}.


\section{Related Works}\label{Sec_Related_Works}

The property of bandwidth adaptivity has been of interest in the literature of regenerating codes for the last few years. There has been a number of researchers who have addressed this problem under different settings. Some, such as \cite{Flexible_Regenerating, BWAdapt_MFR, BWAdapt_Kermarrec, BWAdapt_Opportunistic}, have considered functional repair, in which code design problem reduces to linear network coding, and mainly focused on the theoretical limitations and properties of a bandwidth adaptive regenerating codes. There are, however, a few other works which have considered exact repair as well \cite{Exact_AI_Asymptotic, BAER, BAER_ISIT16, ExplicitMSR}. In this section we briefly review these works and their relevance to the problem setup introduced in the previous section.

%

In the MBR case, \cite{BAER, BAER_ISIT16} addressed a similar setup. For given integers $d_{\min},~d_{\max}$, authors presented a bandwidth adaptive exact repair MBR regenerating code with $D=\{d_{\min}, d_{\min}+1, \cdots, d_{\max}\}$, which is simultaneously error resilient against up to a certain number of erroneous nodes in the network.

For the case of MSR, it was first in \cite{Flexible_Regenerating} that Shah \emph{et. al.} extended the original problem of regenerating code design, introduced in \cite{Regenerating_Conf, Regenerating}, to provide flexibility. In \cite{Flexible_Regenerating} the authors consider the number of participating helpers to be selected independently in each repair or data reconstruction procedure. Moreover, they also allow \emph{asymmetric} participation of helpers, as long as the per-node repair bandwidth of each helper, namely $\beta_{i},~i\in\{1,\cdots, d\}$, is less than a fixed upper bound, $\beta_{\max}$. While the setting considered by Shah \emph{et. al.} provides much more flexibility in repair and reconstruction, they show that the total repair bandwidth, $\gamma$, in their setting should satisfy 
\begin{align}
\gamma \geq \max\{\alpha-\beta_{\max},F\hspace{-2mm}\mod \alpha\} + \left\lfloor \frac{F}{\alpha} \right\rfloor \beta_{\max}, \nonumber
\end{align}
which is always larger than that of the original regenerating codes formulations, except for the MSR case, \emph{i.e.,} $F = k\alpha$, where both settings achieve the same total repair bandwidth. However, the coding scheme presented in \cite{Flexible_Regenerating} only guarantees functional repair.

in \cite{BWAdapt_MFR} Wang \emph{et. al.} considered a functional repair coding scheme with symmetric per-node repair bandwidth, which supports bandwidth adaptivity. Later \cite{BWAdapt_Kermarrec} also considered a similar setup, while in both these works the main focus is on derivation of the optimal trade-off between the storage overhead and repair bandwidth for the functional repair in a \emph{coordinated} setting, where more than one node failure is considered to be repaired together. Note, however, that none of these works address the exact repair with bandwidth adaptivity in regenerating codes. 

Aggrawal \emph{et. al.} \cite{BWAdapt_Opportunistic} also considered a functional repair setup to achieve bandwidth adaptivity in regenerating codes. They analysed the mean-time-to-Data-Loss (MTTDL) in the regenerating codes with and without bandwidth adaptivity. Their analysis is based on a birth-death process model in which the population of available storage node randomly changes with appropriately chosen rates. They showed that bandwidth adaptivity provides a significant gain in terms of MTTDL. 
%

When considering exact repair regenerating codes, Cadambe \emph{et. al.} were the first to address the bandwidth adaptivity \cite{Exact_AI_Asymptotic}. They present an interference alignment based regenerating coding scheme which is capable of performing exact repair with bandwidth adaptivity in the repair procedure. The code presented by Cadambe \emph{et. al.} is the first exact repair regenerating code with bandwidth adaptivity, however, their coding scheme only asymptotically achieves the MSR optimality, when $\alpha$ and $\beta$ tend to infinity with proper ratio. The importance of this result, however, is to show that bandwidth adaptivity could be implemented without extra cost in the optimal trade-off between the storage overhear and repair bandwidth, at least for the MSR codes, even when the exact repair is required. 

In \cite{ProgressiveEngagement}, a similar setup, referred to as \emph{progressive engagement}, is considered for regenerating codes with bandwidth adaptive exact repair, and two coding schemes are provided. While the two coding schemes both preserve the MDS property, none of them satisfy the MSR characteristic equation (\ref{eq_BWAMSR}) simultaneously for different choices of $d$. Another main difference between the progressive engagement setup and the one considered in this work is that the authors in \cite{ProgressiveEngagement} relax the property that \emph{any} subset of surviving nodes could be considered as helpers by assuming that all the remaining systematic nodes are always participating as helpers. Moreover, they require the set of available choices for the number of helpers, namely $D$, to be $D=\{k, k+1, \cdots, n-1\}$, while in our formulation, $D$ does not include all integers between $k$ and $n-1$.

The first explicit constructions for exact repair MSR codes with bandwidth adaptivity and finite subpacketization level were introduced in \cite{ExplicitMSR}. These constructions work for any arbitrary values of $k$ and $n$, and the set of choices for the number of helpers, $D$, could be designed to contain any value $d_{i}$ such that $k < d_{i} < n$. However, the required value for the subpacketization level, $\alpha$, in both constructions is still considerably huge. In particular, for a DSS with $n$ storage nodes, and the set $D=\{d_{1},\cdots,d_{\delta}\}$, constructions in \cite{ExplicitMSR} require
\begin{align}\label{eq_Ye_alpha}
\alpha = \left(\mathrm{lcm}\left(d_{1}-k+1, \cdots, d_{\delta}-k+1 \right)\right)^n.
\end{align}
Therefore, these constructions only achieve optimality for storage of contents which are exponentially large in terms of the number of storage nodes in the system.

In coding schemes suggested for practical networks, the number of nodes, $n$, is usually larger than 10 \cite{FB_Rashmi, EC_for_Azure}, and the required per-node storage should not exceed a few hundreds of megabytes. Indeed at the scale of multiple megabytes, the performance of the system is already limited by bandwidth of the network link and disk I/O \cite{HDFS_Guide, FB_Rashmi}. Unfortunately, the constructions presented in \cite{ExplicitMSR} still struggle to satisfy such choices. For instance the smallest realization of these codes for the setting $n = 14$, and $k = 10$ suggested for Facebook \cite{FB_Rashmi} requires storing more than $10^{15}$ encoded symbols in each storage node, which translates to several peta-bytes per node. Hence the problem of designing exact repair bandwidth adaptive regenerating code remains yet open for practical range of parameters.

In this work, as presented in Theorem \ref{Thm_main}, we address the exact repair bandwidth adaptive MSR code design problem with small subpacketization level. While both \cite{ProgressiveEngagement} and \cite{ExplicitMSR} follow the approach introduced in \cite{Permutation}, which is based on the design of parity check equations, in this work we follow the \emph{product matrix} (PM) framework introduced in \cite{PM_Codes}. Comparing (\ref{eq_Ye_alpha}) with (\ref{eq_alpha_Thm}), one could see that the presented scheme exponentially reduces the required values of $\alpha$ (and $\beta$). However, this scheme works only for $2k-2 \leq d_{i},~\forall{d_{i}\in D}$. As a result, the construction presented in this work only solves the problem in low coding rates. It is worth mentioning that in applications such as general-purpose storage arrays, providing high reliability and fast degraded reads are more important than maximizing the coding rate \cite{NetApp}. However, the design of high-rate bandwidth adaptive MSR codes with small $\alpha$ and $\beta$ still remains a challenging important problem for big data storage systems such as Hadoop.


\section{Coding Scheme}\label{Sec_Coding_Scheme}
In this section we introduce a bandwidth adaptive exact repair MSR coding scheme, which could be designed to provide any required flexibility degree with elements of the set $D$ evenly located between the smallest and the largest element, namely  $d_{1}$ and $d_{\delta}$. We describe the encoding and decoding schemes for storage, repair, and data reconstruction procedures for the adaptive bandwidth MSR code in the following subsections. This coding scheme is closely related to the product matrix MSR code introduced in \cite{PM_Codes}, and could be considered as an extension of the product matrix code that achieves bandwidth adaptivity. As mentioned in the previous section, we will assume all the source and encoded symbols, are elements of a Galois field of an appropriately large field size $q \geq n$, denoted by $\mathbb{F}_{q}$. Moreover, all the operations hereafter are considered to be filed operations of $\mathbb{F}_{q}$. We will refer to $\mathbb{F}_{q}$ as the \emph{code alphabet}.

In the design of the proposed coding scheme, we chose the design parameters, namely $k$, and the required flexibility degree $\delta$. All the other parameters of the code, including $\alpha$, $F$, $D=\{d_{1}, \cdots, d_{\delta} \}$, and $\beta(d_i)$ will be then determined based on $k$, and $\delta$ as follows. The subpacketization level is
\begin{align}\label{eq_alpha_design}
\alpha = (k-1)\cdot \mathrm{lcm}\left(1,2,\cdots, \delta \right).
\end{align}
Moreover, we have $F = k \alpha$, which satisfies the MDS property. Finally, for $D$ we have
\begin{align}\label{eq_D_design}
D=\{d_{1}, \cdots, d_{\delta}\},~~d_{i} = (i+1)(k-1), ~ i \in \{1, \cdots, \delta\}.
\end{align}
and for any $d_{i} \in D$, the associated per-node and total repair bandwidths denoted by $\beta(d_{i})$, and $\gamma(d_{i})$ respectively are
\begin{align}\label{eq_beta_design}
\beta(d_{i}) &= \frac{\alpha}{d_{i}-k+1} = \frac{\alpha}{i (k-1)}, ~~~~\gamma(d_{i}) &= d_{i}\beta(d_{i})=\frac{(i+1)\alpha}{i}.
\end{align}

\subsection{Coding for Storage}\label{SubSec_CodingScheme}
We begin the introduction of the coding scheme by describing the process of encoding the source symbols and deriving the encoded symbols to be stored in the storage nodes. Similar to the product matrix codes, the first step in encoding for storage in this scheme is to arrange the information symbols in a matrix, denoted by $M$, which we refer to hereafter as the \emph{message matrix}. Let
\begin{align}\label{eq_z_define}
z_{\delta}=\mathrm{lcm}\left(1, 2, \cdots, \delta \right).
\end{align}

The message matrix in our coding scheme is structured as follows,
\begin{align}\label{eq_M_define}
M = \left[\begin{array}{c c c c c c c c c c c}
S_{1} ~ & ~ S_{2} ~ & ~ O ~ & ~ O ~ & ~ O ~ & ~ O ~ & & & \cdots & & O \\
S_{2} ~ & ~ S_{3} ~ & ~ S_{4} ~ & ~ O ~ & ~ O ~ & ~ O ~ & & & \cdots & & O \\
O ~ & ~ S_{4} ~ & ~ S_{5} ~ & ~ S_{6} ~ & ~ O ~ & ~ O ~ & & & \cdots & & O \\
O ~ & ~ O ~ & ~ S_{6} ~ & ~ S_{7} ~ & ~ S_{8} ~ & ~ O ~ & & & \cdots & & O \\
\vdots ~ & & & & & & \ddots & & & & \vdots \\
O ~ & & & ~ \cdots & & & & O & S_{2z_{\delta}-4} & S_{2z_{\delta}-3} & S_{2z_{\delta}-2} \\
O ~ & & & ~ \cdots & & & & O & O & S_{2z_{\delta}-2} & S_{2z_{\delta}-1} \\
O ~ & & & ~ \cdots & & & & O & O & O & S_{2z_{\delta}}
\end{array}\right],
\end{align}
where, each $S_{i},~i \in \{1, \cdots, 2z_{\delta}\}$ is a symmetric $(k-1) \times (k-1)$ matrix filled with $k(k-1)/2$ source symbols, and $O$ is a $(k-1) \times (k-1)$ zero matrix. Therefore, $M$ has $(z_{\delta}+1)(k-1)$ rows and $z_{\delta}(k-1)$ columns. Note that the total number of distinct source symbols in the message matrix $M$ is
\begin{align}\label{eq_numel_M}
F = k(k-1)z_{\delta} = k \alpha.
\end{align}

\begin{example}\label{EX_first}
Consider the design parameters $k=3$, and $\delta = 2$. Therefore, from (\ref{eq_z_define}) we have $z_{\delta} = 2$, and using (\ref{eq_numel_M}), the maximum number of source symbols that we can arrange in the message matrix $M$ is $12$. Denoting the source symbols by $s_{1}, \cdots, s_{12}$, then we have
\begin{align}\label{eq_EX_DataMatrix}
M=\left[\begin{array}{c c}
S_{1} & S_{2} \\
S_{2} & S_{3} \\
O & S_{4} 
\end{array}\right] =\left[\begin{array}{c c}
 \vspace{2.5pt}\left[\begin{array}{c c}
s_{1}~ & ~s_{2} \\
s_{2}~ & ~s_{3} \end{array}\right] & \left[\begin{array}{c c} 
s_{4}~ & ~s_{5} \\
s_{5}~ & ~s_{6} \end{array}\right] \\ \vspace{2.5pt}
\left[\begin{array}{c c} 
s_{4}~ & ~s_{5} \\
s_{5}~ & ~s_{6} \end{array}\right] & \left[\begin{array}{c c} 
s_{7}~ & ~s_{8} \\
s_{8}~ & ~s_{9} \end{array}\right] \\
\left[\begin{array}{c c} 
0~~ & ~~0 \\
0~~ & ~~0 \end{array}\right] & \left[\begin{array}{c c} 
s_{10} & s_{11} \\
s_{11} & s_{12} \end{array}\right] 
\end{array} \right].
\end{align}
\end{example}

Once the message matrix is ready, the source encoder creates the vector of coded symbols for each of the $n$ storage nodes, by calculating the product of a node-specific coefficient vector and the message matrix. To describe this process, we first need the following definition.
\begin{defn}\label{def_GVM}[Generalized Vandermonde Matrix]
For distinct and non-zero elements $e_1,\cdots,e_m$ of $\mathbb{F}_q$, and some integer $c \geq 0$, a matrix $A_{m \times \ell}$ with entries 
\begin{align}
A_{i,j} = e_i^{c+j-1}, ~~ \text{for } i \in \{1,\cdots, m\},~j \in \{1, \cdots, \ell\}, \nonumber
\end{align}
is referred to as a \emph{generalized Vandermonde} matrix. 
\end{defn}
The following lemma about these matrices will be used in the proofs of the following theorems,

\begin{lem}\label{Lem_GVM}
Consider distinct and non-zero elements $e_{1}, \cdots, e_{\ell}$ in $\mathbb{F}_{q}$, and an integer $c \geq 0$. Then a square $\ell \times \ell$ generalized Vandermonde matrix, $A$, as defined in Definition \ref{def_GVM}, is invertible in $\mathbb{F}_{q}$.
\end{lem}

The proof of the above lemma, simply follows from the fact that,
\begin{align}
A = \left[\begin{array}{c c c c c}
e_{1}^c~ & ~0~ & \cdots & ~0 \\
0~ & ~e_{2}^c~ & \cdots & ~0 \\
\vdots & ~ ~ & \ddots & \vdots\\
0~ & ~0~  & \cdots & ~e_{\ell}^c \\
\end{array} \right] \left[\begin{array}{c c c c c}
1~ & ~e_{1}~ & ~e_{1}^{2}~ & \cdots & ~e_{1}^{\ell-1} \\
1~ & ~e_{2}~ & ~e_{2}^{2}~ & \cdots & ~e_{2}^{\ell-1} \\
~ & ~~ & \vdots & ~~ & ~ \\
1~ & ~e_{\ell}~ & ~e_{\ell}^{2}~ & \cdots & ~e_{\ell}^{\ell-1}
\end{array}\right], \nonumber
\end{align}
and the two matrices on the right hand side are both full-rank, as one of them is diagonal with non-zero diagonal elements and the other one is a square Vandermonde matrix.

Back to the description of our coding scheme, for distinct and non-zero elements $e_{i}$'s in $\mathbb{F}_{q}$, with $i\in \{1, \cdots, n\}$ we set $c=0$, and define a generalized Vandermonde matrix of size $n\times (z_{\delta}+1)(k-1)$ as
\begin{align}
\Psi = \left[\begin{array}{c c c c c}
1~ & ~e_{1}~ & ~e_{1}^{2}~ & \cdots & ~e_{1}^{(z_{\delta}+1)(k-1)-1} \\
1~ & ~e_{2}~ & ~e_{2}^{2}~ & \cdots & ~e_{2}^{(z_{\delta}+1)(k-1)-1} \\
~ & ~~ & ~~ & \vdots & ~ \\
1~ & ~e_{n}~ & ~e_{n}^{2}~ & \cdots & ~e_{n}^{(z_{\delta}+1)(k-1)-1}
\end{array}\right]. \nonumber
\end{align}
Note that all of the submatrices of $\Psi$ are also generalized Vandermonde matrices. We refer to $\Psi$ as the \emph{coefficient matrix} and denote the $j^{\text{th}}$ row of $\Psi$ by $\underline{\psi}_{j}$. The vector of encoded symbols to be stored on node $j,~j \in \{1, \cdots, n\}$, denoted by $\underline{x}_{j}$, is calculated as
\begin{align}\label{eq_StorageCoding}
\underline{x}_{j}=\underline{\psi}_{j}M.
\end{align}

The vector $\underline{\psi}_{j}$ is the node-specific coefficient vector for storage node $j$. Note that the per-node storage capacity requirement for this coding scheme is then $z_{\delta}(k-1)$ as given by (\ref{eq_alpha_design}).

\begin{example}\label{EX_storage}
Let's consider the setting in Example \ref{EX_first} again, and assume we have $n=7$ nodes in the network. Assume that the code alphabet is the Galois field $\mathbb{F}_{11}$. The coefficient matrix could be formed based on $(e_{1}, e_{2}, \cdots, e_{7})=$ $(1,2, \cdots, 7)$ as follows,
\begin{align}\label{eq_EX_Psi}
\Psi = \left[\begin{array}{c c c c c c}
1~ & ~1~ & ~1~ & ~1~ & ~1~ & ~1 \\
1~ & ~2~ & ~4~ & ~8~ & ~5~ & ~10 \\
1~ & ~3~ & ~9~ & ~5~ & ~4~ & ~1 \\
1~ & ~4~ & ~5~ & ~9~ & ~3~ & ~1 \\
1~ & ~5~ & ~3~ & ~4~ & ~9~ & ~1 \\
1~ & ~6~ & ~3~ & ~7~ & ~9~ & ~10 \\
1~ & ~7~ & ~5~ & ~2~ & ~3~ & ~10
\end{array}\right].
\end{align}

The encoded content of nodes 1 to 7 can be calculated using (\ref{eq_EX_DataMatrix}) and (\ref{eq_StorageCoding}). For instance, for node 1 we have
\begin{align}
\underline{x}_{1} = [x_{1,1}, x_{1,2}, x_{1,3}, x_{1,4}] = [1, 1, 1, 1, 1, 1] M, \nonumber
\end{align}
which gives
\begin{align}
x_{1,1} &= s_{1}+s_{2}+s_{4}+s_{5}, \nonumber \\
x_{1,2} &= s_{2}+s_{3}+s_{5}+s_{6}, \nonumber \\
x_{1,3} &= s_{4}+s_{5}+s_{7}+s_{8}+s_{10}+s_{11}, \nonumber \\
x_{1,4} &= s_{5}+s_{6}+s_{8}+s_{9}+s_{11}+s_{12}, \nonumber
\end{align}
and similarly for node 7 we have
\begin{align}
\underline{x}_{7} = [x_{7,1}, x_{7,2}, x_{7,3}, x_{7,4}] = [1, 7, 5, 2, 3, 10] M, \nonumber
\end{align}
which gives
\begin{align}
x_{7,1} &= s_{1}+7s_{2}+5s_{4}+2s_{5}, \nonumber \\
x_{7,2} &= s_{2}+7s_{3}+5s_{5}+2s_{6}, \nonumber \\
x_{7,3} &= s_{4}+7s_{5}+5s_{7}+2s_{8}+3s_{10}+10s_{11}, \nonumber \\
x_{7,4} &= s_{5}+7s_{6}+5s_{8}+2s_{9}+3s_{11}+10s_{12}. \nonumber
\end{align}
\end{example}

\subsection{Data Reconstruction}

In order to reconstruct all the information stored in the system, the data collector accesses $k$ arbitrary nodes in the network and downloads all their contents. To describe the details of the decoding we use the following lemma.

\begin{lem}\label{Lem_PMreconstruction}
Let $X$ and $\Phi$ be two known matrices of size $k \times (k-1)$, such that $\Phi$ is a generalized Vandermonde matrix, and assume $\Delta$ is a known diagonal matrix of size $k \times k$, with distinct and non-zero diagonal elements. Then the equation 
\begin{align}\label{eq_reconstruction_step}
X = \Phi A + \Delta \Phi B,
\end{align}
is uniquely solvable for unknown $(k-1) \times (k-1)$ symmetric matrices $A$ and $B$.
\end{lem}
 
The proof of this lemma is based on the data reconstruction scheme of the product matrix MSR codes, introduced in \cite{PM_Codes}, and is provided in Appendix \ref{APP_Reconstruction} to help keeping this paper self-contained. The following theorem explains the data reconstruction procedure in this coding scheme.

\begin{thm}\label{Thm_DataReconstruction}
For the coding scheme presented in subsection \ref{SubSec_CodingScheme}, there exists a decoding scheme to reconstruct all the source symbols arranged in the message matrix $M$ from the encoded content of any arbitrary group of $k$ storage nodes.
\end{thm}

\begin{proof}
Let's assume the set of accessed nodes is $\{\ell_{1}, \cdots, \ell_{k}\}$. Moreover, let's denote the $k\times (z_{\delta}+1)(k-1)$ submatrix of $\Psi$ associated with the nodes $\ell_{1}, \cdots, \ell_{k}$, by $\Psi_{\text{DC}}$. We will further denote the submatrix of $\Psi_{\text{DC}}$ consisting of columns $(i-1)(k-1)+1$ through $i (k-1)$, by $\Psi_{\text{DC}}(i)$. In other words, we have a partitioning of $\Psi_{\text{DC}}$'s columns as
\begin{align}\label{eq_Psi_DC_partitioning}
\Psi_{\text{DC}} = \left[\begin{array}{c}
\underline{\psi}_{\ell_{1}} \\
\vdots \\
\underline{\psi}_{\ell_{k}}
\end{array} \right] = \left[\Psi_{\text{DC}}(1), \cdots, \Psi_{\text{DC}}(z_{\delta}+1) \right]. 
\end{align}
As a result, defining the diagonal matrix
\begin{align}
\Lambda_{\text{DC}} = \left[\begin{array}{c c c c c}
e_{\ell_{1}}^{(k-1)}~ & ~0~ & ~0~ & \cdots & ~0 \\
0~ & ~e_{\ell_{2}}^{(k-1)}~ & ~0~ & \cdots & ~0 \\
\vdots ~& ~ & ~ & \ddots & ~\vdots \\
0~ & ~0~ & ~0~ & \cdots & ~e_{\ell_{k}}^{(k-1)} 
\end{array} \right], \nonumber
\end{align}
for each $k\times (k-1)$ submatrix $\Psi_{\text{DC}}(i)$ we have
\begin{align}\label{eq_PsiDC}
\Psi_{\text{DC}}(i+1) = \Lambda_{\text{DC}}\Psi_{\text{DC}}(i).
\end{align}
Similarly, let's denote the matrix consisting of the collected encoded vectors by $X_{\text{DC}}$, and its partitioning to $k\times (k-1)$ submatrices $X_{\text{DC}}(i)$, $i\in\{1,\cdots,z_{\delta}\}$ as follows
\begin{align}\label{eq_X_DC_partitioning}
X_{\text{DC}} = \left[\begin{array}{c}
\underline{x}_{\ell_{1}} \\
\vdots \\
\underline{x}_{\ell_{k}}
\end{array} \right] = \left[ X_{\text{DC}}(1), \cdots, X_{\text{DC}}(z_{\delta}) \right].
\end{align}

The decoding procedure for data reconstruction consists of $z_{\delta}$ consecutive steps. The first step uses only the submatrix $X_{\text{DC}}(1)$. Using (\ref{eq_PsiDC}) we have,
\begin{align}
X_{\text{DC}}(1) &= \left[\Psi_{\text{DC}}(1), \Psi_{\text{DC}}(2) \right]\left[\begin{array}{c}
S_{1} \\
S_{2}
\end{array}\right] \nonumber \\ 
&= \Psi_{\text{DC}}(1)S_{1} + \Lambda_{\text{DC}}\Psi_{\text{DC}}(1)S_{2}. \nonumber
\end{align}

Using Lemma \ref{Lem_PMreconstruction}, the decoder recovers both $S_{1}$, and $S_{2}$, using $X_{\text{DC}}(1)$, in step one. Then, for $i\in\{2, \cdots, z_{\delta}\}$, the decoder performs step $i$ by using submatrix $X_{\text{DC}}(i)$, and decodes submatrices $S_{2i-1}$, and $S_{2i}$, as follows. 

In step $i,~i\in\{2, \cdots, z_{\delta}\}$, of the data reconstruction decoding, the decoder uses the submatrix $X_{\text{DC}}(i)$. Note that 
\begin{align}
X_{\text{DC}}(i) &= \left[\Psi_{\text{DC}}(i-1), \Psi_{\text{DC}}(i), \Psi_{\text{DC}}(i+1) \right]\left[\begin{array}{c}
S_{2(i-1)} \\
S_{2i-1} \\
S_{2i}
\end{array}\right] \nonumber \\ 
&= \Psi_{\text{DC}}(i-1)S_{2(i-1)} + \left[\Psi_{\text{DC}}(i), \Psi_{\text{DC}}(i+1) \right]\left[\begin{array}{c}
S_{2i-1} \\
S_{2i}
\end{array}\right]. \nonumber
\end{align}
Having the submatrix $S_{2(i-1)}$ already recovered from step $i-1$, the decoder first calculates
\begin{align}\label{eq_X_i_calc}
\hat{X}_{\text{DC}}(i) &= X_{\text{DC}}(i) - \Psi_{\text{DC}}(i-1) S_{2(i-1)}. \nonumber \\
&= \left[\Psi_{\text{DC}}(i), \Psi_{\text{DC}}(i+1) \right]\left[\begin{array}{c}
S_{2i-1} \\
S_{2i}
\end{array}\right].
\end{align}
Then from (\ref{eq_PsiDC}), and (\ref{eq_X_i_calc}), we have
\begin{align}
\hat{X}_{\text{DC}}(i) = \Psi_{\text{DC}}(i) S_{2i-1} + \Lambda_{\text{DC}}\Psi_{\text{DC}}(i) S_{2i}. \nonumber
\end{align}
Again using Lemma \ref{Lem_PMreconstruction}, decoder recovers $S_{2i-1}$, and $S_{2i}$ at the end of the step $i$ of the decoding. Therefore, by finishing step $z_{\delta}$, the decoder reconstructs all the submatrices in the message matrix $M$, and recovers the whole data stored in the network.
\end{proof}

\begin{example}\label{EX_DCnotation}
Following the setting described in Example \ref{EX_first} and Example \ref{EX_storage}, we have $k=3$, and the coefficient matrix $\Psi$ is given in (\ref{eq_EX_Psi}). Let's assume the data collector accesses the storage nodes 1, 2, and 4. Then, with $z_{\delta}=2$, we have,
\begin{align}
\Psi_{\text{DC}} = \left[\Psi_{\text{DC}}(1), \Psi_{\text{DC}}(2), \Psi_{\text{DC}}(3)\right]=\left[\begin{array}{c c c c c c}
1~ & ~1~ & ~1~ & ~1~ & ~1~ & ~1 \\
1~ & ~2~ & ~4~ & ~8~ & ~5~ & ~10 \\
1~ & ~4~ & ~5~ & ~9~ & ~3~ & ~1
\end{array} \right], \nonumber
\end{align}
and,
\begin{align}
\Psi_{\text{DC}}(1) = \left[ \begin{array}{c c}
 1~ & ~1 \\
 1~ & ~2 \\
 1~ & ~4
\end{array}\right] \hspace{-1mm}, \Psi_{\text{DC}}(2) = \left[ \begin{array}{c c}
 1~ & ~1 \\
 4~ & ~8 \\
 5~ & ~9
\end{array}\right] \hspace{-1mm},  \Psi_{\text{DC}}(3) = \left[ \begin{array}{c c}
 1~ & ~1 \\
 5~ & ~10 \\
 3~ & ~1
\end{array}\right]. \nonumber
\end{align}

Moreover, we have
\begin{align}\label{eq_Lambda_DC}
\Lambda_{\text{DC}} = \left[\begin{array}{c c c}
 1^2~ & ~0~ & ~0 \\
 0~ & ~2^2~ & ~0 \\
 0~ & ~0~ & ~4^2
\end{array} \right] = \left[\begin{array}{c c c}
 1~ & ~0~ & ~0 \\
 0~ & ~4~ & ~0 \\
 0~ & ~0~ & ~5
\end{array} \right], 
\end{align}
and 
\begin{align}
X_{\text{DC}} = [X_{\text{DC}}(1), X_{\text{DC}}(2)] 
=\left[ \begin{array}{c}
\underline{x}_{1} \\
\underline{x}_{2} \\
\underline{x}_{4} 
\end{array} \right] = \left[\begin{array}{c c c c}
x_{1,1}~ & ~x_{1,2}~ & ~x_{1,3}~ & ~x_{1,4} \\
x_{2,1}~ & ~x_{2,2}~ & ~x_{2,3}~ & ~x_{2,4} \\
x_{4,1}~ & ~x_{4,2}~ & ~x_{4,3}~ & ~x_{4,4}
\end{array} \right], \nonumber
\end{align}
and finally,
\begin{align}\label{eq_EX_XDCs}
X_{\text{DC}}(1)=\left[ \begin{array}{c c}
x_{1,1}~ & ~x_{1,2} \\
x_{2,1}~ & ~x_{2,2} \\
x_{4,1}~ & ~x_{4,2}
\end{array}\right], X_{\text{DC}}(2)=\left[ \begin{array}{c c}
x_{1,3}~ & ~x_{1,4} \\
x_{2,3}~ & ~x_{2,4} \\
x_{4,3}~ & ~x_{4,4}
\end{array}\right].
\end{align}

The first step of decoding in the data reconstruction process based on the encoded data stored in nodes 1, 2 and 4, starts by using only $X_{\text{DC}}(1)$, as given in (\ref{eq_EX_XDCs}). Using Lemma \ref{Lem_PMreconstruction}, The decoder then recovers $S_{1}$, and $S_{2}$ submatrices of the message matrix $M$. 

In the second step then the decoder first calculates 
\begin{align}
\hat{X}_{\text{DC}}(2) = X_{\text{DC}}(2) - \Psi_{\text{DC}}(1)S_{2}, \nonumber
\end{align}
which is equal to 
\begin{align}
\hat{X}_{\text{DC}}(2) = \Psi_{\text{DC}}(2)S_{3} + \Lambda_{\text{DC}}\Psi_{\text{DC}}(2)S_{4}, \nonumber
\end{align}
and hence is of the desired form of (\ref{eq_reconstruction_step}). Therefore, again using lemma \ref{Lem_PMreconstruction}, the decoder recovers $S_{3}$, and $S_{4}$, which completes the data reconstruction.
\end{example}

Algorithm \ref{Alg_DataReconstruction} summarizes the data reconstruction mechanism in this coding scheme.
\begin{algorithm}\label{Alg_DataReconstruction}
\caption{Data Reconstruction}
\begin{algorithmic}[1]
\State Input: $k$, $z_{\delta}$, $\underline{x}_{\ell_{i}}$, $\underline{\psi}_{\ell_{i}}$, and $e_{\ell_{i}}$ for $i\in\{1, \cdots, k\}$.
\State Output: Submatrices $S_{1}, \cdots, S_{2z_{\delta}}$.
\State Form matrices $\Psi_{\text{DC}}(i)$, for $i\in\{1, \cdots, z_{\delta}+1\}$, using (\ref{eq_Psi_DC_partitioning}).
\State Form matrices $X_{\text{DC}}(i)$, for $i\in\{1, \cdots, z_{\delta}\}$, using (\ref{eq_X_DC_partitioning}).
\State Form matrix $\Lambda_{\text{DC}}$, using (\ref{eq_Lambda_DC}).
\State Recover submatrices $S_{1}$, $S_{2}$ from $X_{\text{DC}}(1)$, $\Psi_{\text{DC}}(1)$, and $\Lambda_{\text{DC}}$, using Lemma \ref{Lem_PMreconstruction}.
\For{$i = 2$ to $z_{\delta}$}
\State Calculate $\hat{X}_{\text{DC}}(i) = X_{\text{DC}}(i) - \Psi_{\text{DC}}(i-1)S_{2(i-1)}$.
\State Recover submatrices $S_{2i-1}$, $S_{2i}$ from $\hat{X}_{\text{DC}}(i)$, $\Psi_{\text{DC}}(i)$, and $\Lambda_{\text{DC}}$, using Lemma \ref{Lem_PMreconstruction}.
\EndFor
\end{algorithmic}
\end{algorithm}

\subsection{Bandwidth Adaptive Exact Repair}

We now describe the bandwidth adaptive repair procedure, by assuming that node $f$ is failed and the set of helpers selected for the repair are $\mathcal{H} = \{h_{1}, \cdots, h_{d}\}$, for some $d\in D$. The following theorem describes the repair procedure in this bandwidth adaptive MSR code.

\begin{thm}\label{Thm_BWARepair}
Consider the coding scheme presented in subsection \ref{SubSec_CodingScheme}, with design parameters $k$, and $\delta$, and $D$ as defined in (\ref{eq_D_design}). For any arbitrary failed node $f$, and any arbitrary set of helpers $\mathcal{H}=\{h_{1}, \cdots, h_{d}\}$, for some $d\in D$, there exists a repair scheme for recovering the content of node $f$ with per-node repair bandwidth, 
\begin{align}\label{eq_BWAbeta}
\beta(d) = \frac{\alpha}{d-k+1}.
\end{align}
\end{thm}

\begin{rmrk}
Note that (\ref{eq_D_design}) and (\ref{eq_BWAbeta}) are consistent with (\ref{eq_beta_design}), which satisfies the MSR characteristic equation (\ref{eq_BWAMSR}) for any $d \in D$.
\end{rmrk}

\begin{proof}
Without loss of generality let $d=(m+1)(k-1)$, for some $m\in\{1, \cdots, \delta\}$. Note that (\ref{eq_alpha_design}), and (\ref{eq_D_design}) guarantee that for any $d\in D$, $\alpha$ is an integer multiple of $d-k+1$, hence $\beta(d)$ is an integer. Each helper node $h\in \mathcal{H}$, creates $\beta(d)=\alpha/(d-k+1)$ repair symbols to repair node $f$ as follows. First helper node $h$ partitions its encoded content into $\beta(d)$ equal segments, such that for $i \in \{1,\cdots, \beta(d)\}$, the segment $\underline{x}_{h}(i)$ is of size $\alpha/\beta(d) = d-k+1 = m (k-1)$, and contains elements $x_{h,(i-1)m(k-1)+1}$ through $x_{h,im(k-1)}$. Then we have 
\begin{align}\label{eq_xh_repair_partition}
\underline{x}_{h} = \left[\underline{x}_{h}(1), \cdots, \underline{x}_{h}(\beta(d)) \right].
\end{align}
Similarly, for any node $\ell$, we split the first $\alpha$ entries of a coefficient vector assigned to node $\ell$, namely $\underline{\psi}_{\ell}$, into $\beta(d)$ equal segments as
\begin{align}\label{eq_psi_repair_partition}
\underline{\psi}_{\ell}(1:\alpha) = \left[ \underline{\psi}_{\ell}(1), \cdots, \underline{\psi}_{\ell}(\beta(d)) \right],
\end{align}
where each segment $\underline{\psi}_{\ell}(i)$ is of size $d-k+1 = m (k-1)$.

Now each helper node $h\in \mathcal{H}$, creates its $\beta(d)$ repair symbols as
\begin{align}\label{eq_rh_define}
\underline{r}(h,f) &= \left[ r_{1}(h,f), \cdots, r_{\beta(d)}(h,f)\right] \nonumber \\
&= \left[\underline{x}_{h}(1) \left(\underline{\psi}_{f}(1) \right)^{\intercal}, \cdots, \underline{x}_{h}(\beta(d)) \left(\underline{\psi}_{f}(\beta(d)) \right)^{\intercal} \right].
\end{align}

The repair decoder then stacks $d$ repair vectors $\underline{r}(h,f)$, for $h\in\mathcal{H}$, into a $d\times \beta(d)$ matrix
\begin{align}\label{eq_Upsilon_H}
\Upsilon_{\mathcal{H}} = \left[\begin{array}{c}
\underline{r}(\ell_{1},f) \\
\vdots \\
\underline{r}(\ell_{d},f)
\end{array} \right].
\end{align}
We then introduce the following partitioning of the matrix $\Upsilon_{\mathcal{H}}$, into $\beta(d)$ submatrices, as follows
\begin{align}\label{eq_Upsilon_partition}
\Upsilon_{\mathcal{H}} = \left[ \Upsilon_{\mathcal{H}}(1), \cdots, \Upsilon_{\mathcal{H}}(\beta(d)) \right],
\end{align}
where $\Upsilon_{\mathcal{H}}(i)$, $i\in\{1,\cdots,\beta\}$ is the $i^{\text{th}}$ column in $\Upsilon_{\mathcal{H}}$, of size $d \times 1$.

\begin{figure*}[!h]
\centering
\resizebox{4 in}{!}{
\includegraphics[scale=1]{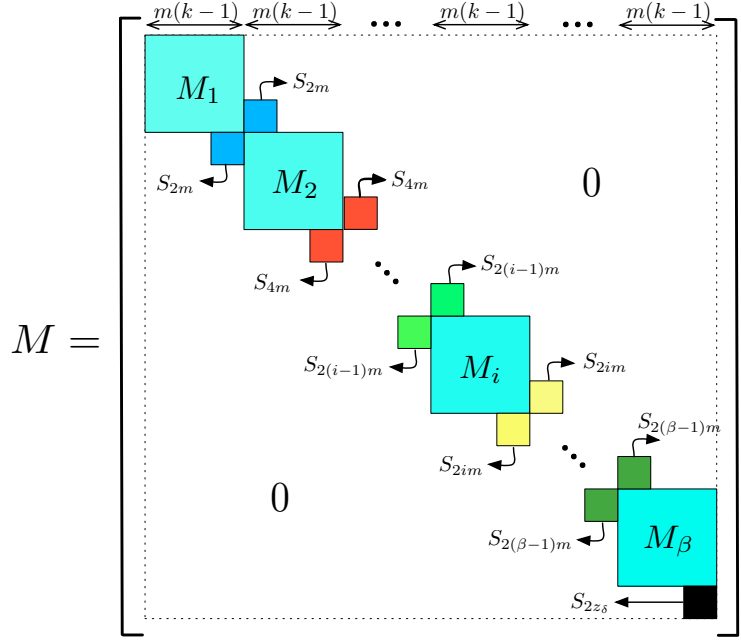}}
\caption{In the above figure $\beta$ represents $\beta(d)$, for some chosen $d \in D$, such that, $d = (m+1)(k-1)$. Moreover, the white area is filled by zeros, and each coloured square represent a non-zero symmetric submatrix of $M$.}
\label{Fig_M_Repair_Partition}
\hrulefill
\end{figure*}

\begin{figure*}[ht]
\begin{align}\label{eq_Mi}
M_{i} = \left[\begin{array}{c c c c c c c c c}
S_{2(i-1)m+1} ~ & ~ S_{2(i-1)m+2} ~ & ~ O ~ & ~ O ~ & ~ O ~ & & \cdots & & ~ O \\
S_{2(i-1)m+2} ~ & ~ S_{2(i-1)m+3} ~ & ~ S_{2(i-1)m+4} ~ & ~ O ~ & ~ O ~ & & \cdots & & ~ O \vspace{5pt}\\
O ~ & ~ S_{2(i-1)m+4} ~ & ~ S_{2(i-1)m+5} ~ & ~ S_{2(i-1)m+6} ~ & ~ O ~ & & \cdots & & ~ O \\
\vdots ~ &  & & & & \ddots & & & \vdots \vspace{5pt}\\
O ~ & & ~ \cdots & & & ~ O ~ & ~ S_{2im-4} ~ & ~ S_{2im-3} ~ & ~ S_{2im-2} \vspace{5pt}\\
O ~ & & ~ \cdots & & & ~ O ~ & ~ O ~ & ~ S_{2im-2} ~ & ~ S_{2im-1}
\end{array} \right].
\end{align}
\hrulefill
\end{figure*}

Before starting to describe the repair decoding procedure, we need to introduce some notations associated to a given repair scenario. Consider a repair procedure with $d=(m+1)(k-1)$. For the corresponding $\beta(d)=\alpha/(d-k+1)$, we will partition matrix $M$ as depicted in Fig. \ref{Fig_M_Repair_Partition}, and equation (\ref{eq_Mi}). Note that this results in $\beta(d)$ non-overlapping diagonal submatices $M_{i}$, $i \in \{1, \cdots, \beta(d)\}$, each of size $m(k-1) \times m(k-1)$, along with $(k-1) \times (k-1)$ symmetric submatrices $S_{2m}, S_{4m}, \cdots, S_{2\beta(d)m}=S_{2z_{\delta}}$ as shown in the figure. It is worth mentioning that the general pattern of the partitioning shown in figure \ref{Fig_M_Repair_Partition} is preserved the same for all $d \in D$, and only the size, and number of the $M_{i}$ diagonal blocks changes for different choices of the parameter $d$. From the construction of the message matrix, introduced in (\ref{eq_M_define}), each $M_{i}$ submatrix will be symmetric. As a result, the message matrix $M$ could be interpreted in terms of the submatrices $M_{i}$, and $S_{2i}$ for $i\in \{1, \cdots, \beta(d)\}$, associated to a repair procedure with $d=(m+1)(k-1)$, $d\in D$.

Finally the last notations we use to describe the adaptive repair decoding scheme, using a given set of helpers $\mathcal{H}=\{h_{1}, \cdots, h_{d}\}$, is,
\begin{align}\label{eq_Omega_i}
\Omega_{\mathcal{H}}(i) = \left[ \begin{array}{c c c c}
e_{h_{1}}^{(i-1)m(k-1)}~ & ~e_{h_{1}}^{(i-1)m(k-1)+1}~ & ~\cdots~ & ~e_{h_{1}}^{(im+1)(k-1)-1}  \\
 ~ & ~~ & \vdots & ~ \\
e_{h_{d}}^{(i-1)m(k-1)}~ & ~e_{h_{d}}^{(i-1)m(k-1)+1}~ & ~\cdots~ & ~e_{h_{d}}^{(im+1)(k-1)-1}
\end{array} \right], ~~i\in \{1, \cdots, \beta(d)\}.
\end{align}
Note that, $\Omega_{\mathcal{H}}(i),~i \in\{1, \cdots, \beta(d)\}$, is a $d \times d$ generalized Vandermonde matrix and hence is invertible as shown in Lemma \ref{Lem_GVM}. We denote the upper $(d-k+1)\times d$ submatrix of $\left(\Omega_{\mathcal{H}}(i)\right)^{-1}$ by $\Theta_{\mathcal{H}}(i)$, and the lower $(k-1) \times d$ submatrix by $\Xi_{\mathcal{H}}(i)$. Then we have,
\begin{align}\label{eq_Omega_partition}
\left(\Omega_{\mathcal{H}}(i)\right)^{-1} = \left[ \begin{array}{c}
\Theta_{\mathcal{H}}(i) \\
\Xi_{\mathcal{H}}(i)
\end{array}\right].
\end{align}

The decoding procedure for the repair of node $f$ is performed in $\beta(d)$ sequential steps as will be described in the following. Let's begin with the first step.

For the failed node $f$, let $\underline{\phi}_{f}$ denote the $1 \times (k-1)$ vector ,
\begin{align}\label{eq_phi_ell}
\underline{\phi}_{f} = \left[1~ , ~e_{f}~, ~\cdots~, ~e_{f}^{k-2} \right].
\end{align}
Using a partitioning similar to (\ref{eq_xh_repair_partition}) for $\underline{x}_{f}$, then we have,
\begin{align}\label{eq_x_f_1}
\underline{x}_{f}(1) &= \underline{\psi}_{f} \left[\begin{array}{c c c c}
\block(4,4){\left[\begin{array}{c c c c}
\vspace{7pt} ~~~~ & ~~ & ~~ & ~~~~ \\
 ~~~ & ~~~ & M_{1}~ & ~~~~~~ \\
\vspace{7pt} ~~~~ & ~~ & ~~ & ~~~~ 
\end{array}\right]} \\
&&& \\
&&& \\
&&& \\
~O~ & ~\cdots ~ & ~O & S_{2m}
\end{array} \right]  \nonumber \\
&= \underline{\psi}_{f}(1) M_{1} + \left[O, \cdots, O, e_{f}^{m(k-1)}\underline{\phi}_{f}S_{2m} \right]_{(k-1) \times m(k-1)}.
\end{align}
In the first step, the decoder recovers the two terms on the right in the above equation to reconstruct $\underline{x}_{f}(1)$ using only the first repair symbol received from each of the helpers, namely $r_{1}(h_{i},f)$, for $i\in \{1, \cdots, d\}$, as follows.

Using (\ref{eq_xh_repair_partition}) to (\ref{eq_Upsilon_partition}), and the partitioning denoted in Fig. \ref{Fig_M_Repair_Partition}, the submatrix $\Upsilon_{\mathcal{H}}(1)$, introduced in (\ref{eq_Upsilon_partition}) can be written as,
\begin{align}\label{eq_Upsilon_1_rewrite}
\Upsilon_{\mathcal{H}}(1) &= \left[\begin{array}{c}
\underline{x}_{h_{1}}(1)\left(\underline{\psi}_{f}(1)\right)^{\intercal} \\
\vdots \\
\underline{x}_{h_{d}}(1)\left(\underline{\psi}_{f}(1)\right)^{\intercal}
\end{array}\right] = \left[\begin{array}{c}
\underline{x}_{h_{1}}(1) \\
\vdots \\
\underline{x}_{h_{d}}(1)
\end{array}\right] \left(\underline{\psi}_{f}(1)\right)^{\intercal} \nonumber \\
&= \left[  \begin{array}{c c c c c}
1~ & ~e_{h_{1}}~ & ~e_{h_{1}}^{2}~ & ~\cdots~ & ~e_{h_{1}}^{(m+1)(k-1)-1}  \\
 ~ & ~~ & ~~ & \vdots & ~ \\
1~ & ~e_{h_{d}}~ & ~e_{h_{d}}^{2}~ & ~\cdots~ & ~e_{h_{d}}^{(m+1)(k-1)-1} 
\end{array} \right] \left[\begin{array}{c c c c}
\block(4,4){\left[\begin{array}{c c c c}
\vspace{7pt} ~~~~ & ~~ & ~~ & ~~~~ \\
 ~~~ & ~~~ & M_{1}~ & ~~~~~~ \\
\vspace{7pt} ~~~~ & ~~ & ~~ & ~~~~ 
\end{array}\right]} \\
&&& \\
&&& \\
&&& \\
~O~ & ~\cdots ~ & ~O & S_{2m}
\end{array} \right] \left( \underline{\psi}_{f}(1)\right)^{\intercal}.
\end{align}
Then using (\ref{eq_Omega_i}), we have, 
\begin{align}\label{eq_repair_step1}
\Upsilon_{\mathcal{H}}(1) = \Omega_{\mathcal{H}}(1) \left[\begin{array}{c c c c}
\block(4,4){\left[\begin{array}{c c c c}
\vspace{7pt} ~~~~ & ~~ & ~~ & ~~~~ \\
 ~~~ & ~~~ & M_{1}~ & ~~~~~~ \\
\vspace{7pt} ~~~~ & ~~ & ~~ & ~~~~ 
\end{array}\right]} \\
&&& \\
&&& \\
&&& \\
~O~ & ~\cdots ~ & ~O & S_{2m}
\end{array} \right] \left( \underline{\psi}_{f}(1)\right)^{\intercal}. 
\end{align}
Multiplying the inverse of $\Omega_{\mathcal{H}}(1)$ from right to the both sides of (\ref{eq_repair_step1}), and using (\ref{eq_Omega_partition}) the decoder derives
\begin{align}
\left[\begin{array}{c c c c}
\block(4,4){\left[\begin{array}{c c c c}
\vspace{7pt} ~~~~ & ~~ & ~~ & ~~~~ \\
 ~~~ & ~~~ & M_{1}~ & ~~~~~~ \\
\vspace{7pt} ~~~~ & ~~ & ~~ & ~~~~ 
\end{array}\right]} \\
&&& \\
&&& \\
&&& \\
~O~ & ~\cdots ~ & ~O & S_{2m}
\end{array} \right] \left( \underline{\psi}_{f}(1)\right)^{\intercal} = \left[ \begin{array}{c}
\Theta_{\mathcal{H}}(i) \\
\Xi_{\mathcal{H}}(i)
\end{array}\right] \Upsilon_{\mathcal{H}}(1) \nonumber
\end{align}
That gives,
\begin{align}\label{eq_M1_repair}
M_{1} \left( \underline{\psi}_{f}(1) \right)^{\intercal} = \Theta_{\mathcal{H}}(1) \Upsilon_{\mathcal{H}}(1),
\end{align}
and similarly, using (\ref{eq_phi_ell}),
\begin{align}\label{eq_S2m_repair}
S_{2m} \left( e_{f}^{(m-1)(k-1)} \underline{\phi}_{f} \right)^{\intercal} = \Xi_{\mathcal{H}}(1) \Upsilon_{\mathcal{H}}(1).
\end{align} 
Since both $M_{1}$, and $S_{2m}$ are symmetric, from (\ref{eq_M1_repair}) we have,
\begin{align}\label{eq_additiveterm1}
\underline{\psi}_{f}(1) M_{1} = \left(\Theta_{\mathcal{H}}(1) \Upsilon_{\mathcal{H}}(1)\right)^{\intercal},
\end{align}
and from (\ref{eq_S2m_repair}), by multiplying the scalar $e_{f}^{(k-1)}$, we get
\begin{align}\label{eq_additiveterm2}
e_{f}^{m(k-1)}\underline{\phi}_{f} S_{2m}  = e_{f}^{(k-1)} \left( \Xi_{\mathcal{H}}(1) \Upsilon_{\mathcal{H}}(1) \right)^{\intercal}.
\end{align} 

From (\ref{eq_additiveterm1}), and (\ref{eq_additiveterm2}), and using (\ref{eq_x_f_1}) the decoder then recovers $\underline{x}_{f}(1)$ as,
\begin{align}\label{eq_xf1_repair}
\underline{x}_{f}(1) = \underline{\psi}_{f}(1) M_{1} + \left[O, \cdots, O, e_{f}^{m(k-1)}\underline{\phi}_{f}S_{2m} \right]_{(k-1) \times m(k-1)}, 
\end{align}
where, the rightmost term in the above expression is derived by padding $m-1$, $(k-1) \times (k-1)$ zero matrices, $O$, to the left of the matrix calculated in (\ref{eq_additiveterm2}).

In step $i$ for $i=2$ through $\beta(d)$ of the repair decoding, the decoder then recovers $\underline{x}_{f}(i)$, using $\Upsilon_{\mathcal{H}}(i)$ received from the helpers, along with $e_{f}^{(i-1)m(k-1)} \underline{\phi}_{f} S_{2(i-1)m}$, recovered from the step $i-1$ of decoding. To this end, first note that similar to (\ref{eq_Upsilon_1_rewrite}) the repair symbols in $\Upsilon_{\mathcal{H}}(i)$ can be written as,
\begin{align}
\Upsilon_{\mathcal{H}}(i) &= \left[\begin{array}{c}
\underline{x}_{h_{1}}(i)\left(\underline{\psi}_{f}(i)\right)^{\intercal} \\
\vdots \\
\underline{x}_{h_{d}}(i)\left(\underline{\psi}_{f}(i)\right)^{\intercal}
\end{array}\right] = \left[\begin{array}{c}
\underline{x}_{h_{1}}(i) \\
\vdots \\
\underline{x}_{h_{d}}(i)
\end{array}\right] \left(\underline{\psi}_{f}(i)\right)^{\intercal} \nonumber \\
&= \left[  \begin{array}{c c c c}
e_{h_{1}}^{(i-1)m(k-1)}~ & ~e_{h_{1}}^{(i-1)m(k-1)+1}~ & ~\cdots~ & ~e_{h_{1}}^{(im+1)(k-1)-1}  \\
 ~ & ~~ & \vdots & ~ \\
e_{h_{d}}^{(i-1)m(k-1)}~ & ~e_{h_{d}}^{(i-1)m(k-1)+1}~ & ~\cdots~ & ~e_{h_{d}}^{(im+1)(k-1)-1}
\end{array} \right] \left[\begin{array}{c c c c}
S_{\tiny{2(i-1)m}} & O & \cdots & O \\
\block(4,4){\left[\begin{array}{c c c c}
\vspace{7pt} ~~~~~~ & ~~ & ~~ & ~~~~~~ \\
 ~~~~~ & ~~~ & M_{i}~ & ~~~~~~~~ \\
\vspace{7pt} ~~~~~~ & ~~ & ~~ & ~~~~~~ 
\end{array}\right]} \\
&&& \\
&&& \\
&&& \\
O & ~\cdots ~ & ~O & S_{2m}
\end{array} \right] \left( \underline{\psi}_{f}(i)\right)^{\intercal}. \nonumber 
\end{align}
 Using (\ref{eq_Omega_i}) we can rewrite the above equation as,
 \begin{align}
 \Upsilon_{\mathcal{H}}(i) = \left[\begin{array}{c}
e_{\ell_{1}}^{(i-1)m(k-1) -\mu} \underline{\phi}_{\ell_{1}} \\
\vdots \\
e_{\ell_{d}}^{(i-1)m(k-1) -\mu} \underline{\phi}_{\ell_{d}}
\end{array} \right] \hspace{-1mm} S_{2(i-1)m} \hspace{-1mm} \left(e_{f}^{(i-1)m(k-1)}\underline{\phi}_{f}\right)^{\intercal} + \Omega_{\mathcal{H}}(i) \left[\begin{array}{c c c c}
\block(4,4){\left[\begin{array}{c c c c}
\vspace{7pt} ~~~~ & ~~ & ~~ & ~~~~ \\
 ~~~ & ~~~ & M_{i}~ & ~~~~~~ \\
\vspace{7pt} ~~~~ & ~~ & ~~ & ~~~~ 
\end{array}\right]} \\
&&& \\
&&& \\
&&& \\
~O~ & ~\cdots ~ & ~O & S_{2im}
\end{array} \right] \left( \underline{\psi}_{f}(i)\right)^{\intercal}. \nonumber
\end{align}
The decoder first removes the contribution of the $S_{2(i-1)m}$ submatrix in the repair symbols in $\Upsilon_{\mathcal{H}}(i)$ by calculating 
\begin{align}\label{eq_UpsilonHat_i}
\hat{\Upsilon}_{\mathcal{H}}(i) \hspace{-1mm} = \hspace{-1mm} \Upsilon_{\mathcal{H}}(i) \hspace{-1mm} - \hspace{-1mm} \left[\begin{array}{c}
e_{\ell_{1}}^{(i-1)m(k-1)-(k-1)} \underline{\phi}_{\ell_{1}} \\
\vdots \\
e_{\ell_{d}}^{(i-1)m(k-1)-(k-1)} \underline{\phi}_{\ell_{d}}
\end{array} \right] \hspace{-1mm} S_{2(i-1)m} \hspace{-1mm} \left(e_{f}^{(i-1)m(k-1)}\underline{\phi}_{f}\right)^{\intercal}.
\end{align}
In the above expression, $S_{2(i-1)m} \left(e_{f}^{(i-1)m(k-1)}\underline{\phi}_{f}\right)^{\intercal}$ is itself derived by transposing $e_{f}^{(i-1)m(k-1)} \underline{\phi}_{f} S_{2(i-1)m}$. As a result we have,
\begin{align}
\hat{\Upsilon}_{\mathcal{H}}(i) = \Omega_{\mathcal{H}}(i) \left[\begin{array}{c c c c}
\block(4,4){\left[\begin{array}{c c c c}
\vspace{7pt} ~~~~ & ~~ & ~~ & ~~~~ \\
 ~~~ & ~~~ & M_{i}~ & ~~~~~~ \\
\vspace{7pt} ~~~~ & ~~ & ~~ & ~~~~ 
\end{array}\right]} \\
&&& \\
&&& \\
&&& \\
~O~ & ~\cdots ~ & ~O & S_{2im}
\end{array} \right] \left( \underline{\psi}_{f}(i)\right)^{\intercal}. \nonumber
\end{align}
Therefore, similar to (\ref{eq_M1_repair}) through (\ref{eq_additiveterm2}) the decoder derives,
\begin{align}\label{eq_Mi_repair}
\underline{\psi}_{f}(i)M_{i} = \left(\Theta_{\mathcal{H}}(i) \Upsilon_{\mathcal{H}}(i)\right)^{\intercal},
\end{align}
and
\begin{align}\label{eq_additiveterm2i}
e_{f}^{im(k-1)} \underline{\phi}_{f} S_{2im} = e_{f}^{(k-1)} \left( \Xi_{\mathcal{H}}(i) \Upsilon_{\mathcal{H}}(i)\right)^{\intercal}.
\end{align} 
Finally, using (\ref{eq_Mi_repair}) and (\ref{eq_additiveterm2i}), we have
\begin{align}\label{eq_xf_i}
\underline{x}_{f}(i) = \underline{\psi}_{f}(i) M_{i} + \left[O, \cdots, O, e_{f}^{im(k-1)}\underline{\phi}_{f}S_{2im} \right]_{(k-1) \times m(k-1)}.
\end{align} 
\end{proof}

The following algorithm summarizes the bandwidth adaptive repair procedure.
\begin{algorithm}\label{Alg_AdaptiveRepair}
\caption{Bandwidth Adaptive Repair}
\begin{algorithmic}[1]
\State Input: $f$, $e_{f}$, $\beta$, $m$, $k$ and $\underline{r}(h,f),~\underline{\psi}_{h},~e_{h}$, for $h\in \mathcal{H}$.
\State Form matrices $\Upsilon_{\mathcal{H}}(i)$, for $i\in\{1, \cdots, \beta \}$, using (\ref{eq_Upsilon_H}), and (\ref{eq_Upsilon_partition}).
\State Form vectors $\underline{\phi}_{f}$ using (\ref{eq_phi_ell}).
\State Calculate matrices $\Omega_{\mathcal{H}}(i)$, and $\left(\Omega_{\mathcal{H}}(i)\right)^{-1}$, for $i \in \{1, \cdots, \beta\}$, using (\ref{eq_Omega_i}).
\State Calculate matrices $\underline{\psi}_{f}(1) M_{1}$ using (\ref{eq_additiveterm1}).
\State Calculate matrices $e_{f}^{ma}\underline{\phi}_{f} S_{2m}$ using (\ref{eq_additiveterm2}).
\State Recover $\underline{x}_{f}(1)$ using (\ref{eq_xf1_repair}).
\For{$i = 2$ to $\beta$}
\State Calculate $\hat{\Upsilon}_{\mathcal{H}}(i)$ using (\ref{eq_UpsilonHat_i}).
\State Calculate $\underline{\psi}_{f}(i)M_{i}$ using (\ref{eq_M1_repair}).
\State Calculate $e_{f}^{im(k-1)} \underline{\phi}_{f} S_{2im}$ using (\ref{eq_additiveterm2i}).
\State Recover $\underline{x}_{f}(i)$ using (\ref{eq_xf_i}).
\EndFor
\State Form $\underline{x}_{f} = [\underline{x}_{f}(1), \cdots, \underline{x}_{f}(\beta)]$.
\end{algorithmic}
\end{algorithm}

\begin{rmrk}
Note that calculating the inverse matrices $\left(\Omega_{\mathcal{H}}(i)\right)^{-1}$, for $i \in \{1, \cdots, \beta\}$, can be carried out recursively since from (\ref{eq_Omega_i}) we have,
\begin{align}
\Omega_{\mathcal{H}}(i) = \left[ \begin{array}{c c c c c}
e_{h_{1}}^{(i-1)m(k-1)}~ & ~0~ & \cdots & ~0 \\
0~ & ~e_{h_{2}}^{(i-1)m(k-1)}~ & \cdots & ~0 \\
\vdots & ~~ & \ddots & ~~ & \vdots \\
0~ & ~0~ & \cdots & ~e_{h_{d}}^{(i-1)m(k-1)}
\end{array}\right]\Omega_{\mathcal{H}}(1) \nonumber
\end{align}
\end{rmrk}

\begin{rmrk}
In a DSS with $n$ nodes, for $D=\{d_{1},\cdots, d_{\delta}\}$, the bandwidth adaptive MSR codes presented in \cite{ExplicitMSR} require 
\begin{align}\label{eq_Ye_Barg_alpha}
\alpha = \left(\mathrm{lcm}\left(d_{1}-k+1, \cdots, d_{\delta}-k+1 \right)\right)^n. 
\end{align}
Comparing (\ref{eq_Ye_Barg_alpha}) with (\ref{eq_alpha_design}), one could see that the presented scheme reduces the required $\alpha$ (and $\beta$) values exponentially. However, this scheme works only for $2k-2 \leq d_{i},~\forall{d_{i}\in D}$. Hence, the design of high-rate bandwidth adaptive MSR codes with small $\alpha$ and $\beta$ still remains an open problem.
\end{rmrk}

The following example provides a detailed illustration of the MSR bandwidth adaptive exact repair procedure in the same setup as described in Examples \ref{EX_first} to \ref{EX_DCnotation}. 

\begin{example}\label{EX_BWAExactRepair}
As in the previous examples we will consider $k=3$, and $\delta=2$, and using (\ref{eq_alpha_design}), $\alpha = 4$. The code alphabet is $\mathbb{F}_{11}$, and the message matrix $M$ as given in (\ref{eq_EX_DataMatrix}). As a result, using (\ref{eq_D_design}), we have $D=\{d_{1},d_{2}\}=\{4,6\}$, and from (\ref{eq_BWAbeta}) their associated per-node repair bandwidths are $\beta_{1}=2$, and $\beta_{2}=1$. Without loss of generality, let's assume node 7 is failed, \emph{i.e.} $f=7$. In this setup, the following two repair scenarios are then possible.

\emph{A:} One option is to use $d=d_{2}=6$ helpers, and download only $\beta_{2} = 1$ repair symbol from each of them, which means $d=(m+1)(k-1)$ for $m=2$. Consider $\mathcal{H}=\{1,2,\cdots, 6\}$. In this case, using the coefficient matrix $\Psi$ for this setup, which is given in (\ref{eq_EX_Psi}), from (\ref{eq_psi_repair_partition}), with $\alpha=4, \beta=1$, for any node $\ell$ we have
\begin{align}
\underline{\psi}_{\ell}(1)=\underline{\psi}_{\ell}(1:4), \nonumber
\end{align}
which is a row vector consisting of the first $\alpha=4$ elements of the coefficient vector assigned to node $\ell$. As a result, using (\ref{eq_xh_repair_partition}), each helper node $h \in \mathcal{H}$ will use all of its encoded content $\underline{x}_{h}$ to create a single repair symbol as
\begin{align}
\underline{r}(h,f) = [r_{1}(h,f)] &= \left[\underline{x}_{h}(1)\cdot \left(\underline{\psi}_{f}(1)\right)^{\intercal}\right] \nonumber \\
&= \left[\underline{x}_{h}\cdot [1, 7, 5, 2]^{\intercal}\right]. \nonumber
\end{align}

The repair decoder will then receive 
\begin{align}
\Upsilon_{\mathcal{H}} &= [\Upsilon_{\mathcal{H}}(1)] \nonumber \\
&= \left[\begin{array}{c}
\underline{r}(1,f) \\
\vdots \\
\underline{r}(6,f)
\end{array}\right] = \left[\begin{array}{c}
x_{1,1}+7x_{1,2}+5x_{1,3}+2x_{1,4} \\
x_{2,1}+7x_{2,2}+5x_{2,3}+2x_{2,4} \\
x_{3,1}+7x_{3,2}+5x_{3,3}+2x_{3,4} \\
x_{4,1}+7x_{4,2}+5x_{4,3}+2x_{4,4} \\
x_{5,1}+7x_{5,2}+5x_{5,3}+2x_{5,4} \\
x_{6,1}+7x_{6,2}+5x_{6,3}+2x_{6,4}
\end{array} \right]. \nonumber
\end{align}

Moreover, using (\ref{eq_phi_ell}), and with $e_{f} = 7$ from (\ref{eq_EX_Psi}) as also used in all previous examples, for $f=7$ we get
\begin{align}
\underline{\phi}_{f} = \left[1, 7 \right]. \nonumber
\end{align}

Finally notice that in this case the repair decoding will only have one single step. Therefore, for step $i=1$, with $m=2$, and $k=3$ we get $e_{\ell}^{im(k-1)} = e_{\ell}^{4}$, and hence,
\begin{align}\label{eq_EX_Omega_caseOne}
\Omega_{\mathcal{H}}(1) = \left[ \begin{array}{c c c c c c}
1~ & ~1~ & ~1~ & ~1~ & ~1~ & ~1 \\
1~ & ~2~ & ~4~ & ~8~ & ~5~ & ~10 \\
1~ & ~3~ & ~9~ & ~5~ & ~4~ & ~1 \\
1~ & ~4~ & ~5~ & ~9~ & ~3~ & ~1 \\
1~ & ~5~ & ~3~ & ~4~ & ~9~ & ~1 \\
1~ & ~6~ & ~3~ & ~7~ & ~9~ & ~10
\end{array}\right].
\end{align}

Then, based on the partition represented in Fing. (\ref{Fig_M_Repair_Partition}), for 
\begin{align}
M_{1} = \left[\begin{array}{c c}
S_{1}~ & ~S_{2} \\
S_{2}~ & ~S_{3}
\end{array} \right], \nonumber
\end{align}
the decoder has access to 
\begin{align}
\Upsilon_{\mathcal{H}}&=\Omega_{\mathcal{H}}(1) \left[\begin{array}{c c c}
\block(3,3){\left[\begin{array}{c c c}
 ~ & ~ & ~  \\ 
 ~ & M_{1} & ~ \\
 ~ & ~ & ~
\end{array}\right]} \\
&& \\
&& \\
~~O && S_{4}~~
\end{array} \right] \left( \underline{\psi}_{f}(1)\right)^{\intercal} \nonumber \\
&=\Omega_{\mathcal{H}}(1)\left[\begin{array}{c c}
\block(2,2){\left[\begin{array}{c c c c} 
s_{1}~ & ~s_{2}~ & ~s_{4}~ & ~s_{5} \\
s_{2}~ & ~s_{3}~ & ~s_{5}~ & ~s_{6} \\
s_{4}~ & ~s_{5}~ & ~s_{7}~ & ~s_{8} \\
s_{5}~ & ~s_{6}~ & ~s_{8}~ & ~s_{9}
\end{array}\right]} \\
& \\
& \\
& \\  
\hspace{1 mm}\left[\begin{array}{c c}
0~ & ~~0 \\
0~ & ~~0
\end{array}\right] & \hspace{-2mm}\left[\begin{array}{c c}
s_{10} & ~s_{11} \\
s_{11} & ~s_{12}
\end{array} \right]
\end{array} \right] \left( \underline{\psi}_{f}(1)\right)^{\intercal}. \nonumber
\end{align}

Note that the $\Omega_{\mathcal{H}}(1)$, given in (\ref{eq_EX_Omega_caseOne}) is an invertible matrix in the code alphabet $\mathbb{F}_{11}$, and we have
\begin{align}
\left(\Omega_{\mathcal{H}}(1)\right)^{-1} = \left[\begin{array}{c}
\Theta_{\mathcal{H}}(1) \\
\Xi_{\mathcal{H}}(1)
\end{array}\right] = \left[\begin{array}{c}
\left[ \begin{array}{c c c c c c}
6~ & ~7~ & ~9~ & ~7~ & ~6~ & ~10 \\
10~ & ~10~ & ~9~ & ~0~ & ~3~ & ~1 \\
3~ & ~6~ & ~9~ & ~1~ & ~9~ & ~5 \\
1~ & ~8~ & ~0~ & ~8~ & ~2~ & ~3 
\end{array}\right] \\
\vspace{1.5 mm}\left[\begin{array}{c c c c c c}
~2~ & ~~7~ & ~7~ & ~5~ & ~8~ & ~4 \\
~1~ & ~~6~ & ~10~ & ~1~ & ~5~ & ~10
\end{array}\right]
\end{array} \right]. \nonumber
\end{align}

The decoder then calculates the lost data $\underline{x}_{f}$, using (\ref{eq_xf1_repair}), with $e_{f}^{(k-1)}=7^2=5$ in $\mathbb{F}_{11}$, as follows
\begin{align}
\underline{x}_{f} = \left(\Theta_{\mathcal{H}}(1)\Upsilon_{\mathcal{H}}\right)^{\intercal}+\left[0,0, 5\left(\Xi_{\mathcal{H}}(1)\Upsilon_{\mathcal{H}}\right)^{\intercal}\right]. \nonumber
\end{align}

\emph{B:} The second option is to use $d=d_{1}=4$ helpers. With $k=3$ we have $d = (m+1)(k-1)$ for $m=1$. In this case, using (\ref{eq_BWAbeta}) we will have $\beta_{1}=4/(4-2)=2$ repair symbols per helper node. Let's without loss of generality assume node 7 is failed, \emph{i.e.} $f=7$, and $\mathcal{H}=\{1,2,3,4\}$ is the set of helper nodes chosen to perform the repair. As a result, using the coefficient matrix in the code alphabet $\mathbb{F}_{11}$, given in (\ref{eq_EX_Psi}), from equation (\ref{eq_psi_repair_partition}) for the helper nodes we have
\begin{align}
\underline{\psi}_{1}(1) = [1, 1], ~~ \underline{\psi}_{1}(2) = [1, 1], \nonumber \\
\underline{\psi}_{2}(1) = [1, 2], ~~ \underline{\psi}_{2}(2) = [4, 8], \nonumber \\
\underline{\psi}_{3}(1) = [1, 3], ~~ \underline{\psi}_{3}(2) = [9, 5], \nonumber \\
\underline{\psi}_{4}(1) = [1, 4], ~~ \underline{\psi}_{4}(2) = [5, 9], \nonumber
\end{align}
and for the failed node $f=7$, 
\begin{align}
\underline{\psi}_{f}(1) = [1, 7], ~~ \underline{\psi}_{f}(2) = [5, 2]. \nonumber
\end{align}
Similarly, for the coded content of each of these nodes we consider the following partition
\begin{align}
\underline{x}_{\ell} &= [\underline{x}_{\ell}(1), \underline{x}_{\ell}(2)] \nonumber \\
&= [[x_{\ell,1}, x_{\ell,2}], [x_{\ell,3}, x_{\ell,4}]], ~~\text{for}~\ell\in \{1,2,3,4,7\}. \nonumber
\end{align}

Each helper node $h\in\mathcal{H}$ then creates two repair symbols 
\begin{align}
\underline{r}(h,f) &= [r_{1}(h,f),r_{2}(h,f)] \nonumber \\
&= \left[\underline{x}_{\ell}(1)\left(\underline{\psi}_{f}(1) \right)^{\intercal}, \underline{x}_{\ell}(2)\left(\underline{\psi}_{f}(2) \right)^{\intercal}\right] \nonumber \\ &=[x_{h,1}+7x_{h,2},~5x_{h,3}+2x_{h,4}], \nonumber
\end{align}
and the repair decoder receives
\begin{align}
\Upsilon_{\mathcal{H}} &= [\Upsilon_{\mathcal{H}}(1), \Upsilon_{\mathcal{H}}(2)] \nonumber \\
&=\left[\begin{array}{cc}
\left[\begin{array}{c}
r_{1}(1,f) \\
\vdots \\
r_{1}(4,f)
\end{array}\right] & \left[\begin{array}{c}
r_{2}(1,f) \\
\vdots \\
r_{2}(4,f)
\end{array} \right]
\end{array} \right] \nonumber \\
&= \left[ \begin{array}{c c}
\left[\begin{array}{c}
x_{1,1}+7x_{1,2} \\
x_{2,1}+7x_{2,2} \\
x_{3,1}+7x_{3,2} \\
x_{4,1}+7x_{4,2}
\end{array} \right] & \left[ \begin{array}{c}
5x_{1,3}+2x_{1,4} \\
5x_{2,3}+2x_{2,4} \\
5x_{3,3}+2x_{3,4} \\
5x_{4,3}+2x_{4,4}
\end{array} \right]
\end{array} \right]. \nonumber
\end{align}

Moreover, using (\ref{eq_phi_ell}), and with $e_{f} = 7$ from (\ref{eq_EX_Psi}) as also used in all previous examples, for $f=7$ we get
\begin{align}
\underline{\phi}_{f} = \left[1, 7 \right]. \nonumber
\end{align}

Finally notice that in this case the repair decoding will have two steps. Therefore, for step $i=1$, with $m=1$, and $k=3$ we get $e_{\ell}^{im(k-1)} = e_{\ell}^{2}$, and hence, using equation (\ref{eq_Omega_i}), we have
\begin{align}\label{eq_EX_Omega_caseTwo_1}
\Omega_{\mathcal{H}}(1) = \left[ \begin{array}{c c c c}
1~ & ~1~ & ~1~ & ~1 \\
1~ & ~2~ & ~4~ & ~8 \\
1~ & ~3~ & ~9~ & ~5 \\
1~ & ~4~ & ~5~ & ~9
\end{array}\right].
\end{align}

Then, according to the partition represented in Fig. (\ref{Fig_M_Repair_Partition}), for 
\begin{align}
M_{1} = S_{1}, ~~M_{2} = S_{3}, \nonumber
\end{align}
 the repair decoder has access to 
\begin{align}
\Upsilon_{\mathcal{H}}(1)&=\Omega_{\mathcal{H}}(1) \left[\begin{array}{c}
 M_{1}  \\
 S_{2}
\end{array} \right] \left( \underline{\psi}_{f}(1)\right)^{\intercal} \nonumber \\
&=\Omega_{\mathcal{H}}(1)\left[\begin{array}{c}
\left[\begin{array}{c c}
s_{1}~ & ~s_{2} \\
s_{2}~ & ~s_{3}
\end{array}\right] \\
\left[\begin{array}{c c}
s_{4}~ & ~s_{5} \\
s_{5}~ & ~s_{6}
\end{array} \right]
\end{array} \right] \left[\begin{array}{c}
1 \\
7
\end{array} \right]. \nonumber
\end{align}
As expected $\Omega_{\mathcal{H}}(1)$ is an invertible matrix in the code alphabet $\mathbb{F}_{11}$, and we have
\begin{align}
\left(\Omega_{\mathcal{H}}(1)\right)^{-1} = \left[\begin{array}{c}
\Theta_{\mathcal{H}}(1) \\
\Xi_{\mathcal{H}}(1)
\end{array}\right] = \left[\begin{array}{c}
\left[ \begin{array}{c c c c}
4~ & ~5~ & ~4~ & ~10 \\
3~ & ~4~ & ~4~ & ~0 
\end{array}\right] \\
\vspace{2mm}\left[\begin{array}{c c c c}
7~ & ~7~ & ~9~ & ~10 \\
9~ & ~6~ & ~5~ & ~2 
\end{array}\right]
\end{array} \right]. \nonumber
\end{align}

The decoder then calculates the lost data $\underline{x}_{f}(1)$, using (\ref{eq_xf1_repair}), with $e_{f}^{m(k-1)}=7^2=5$ in $\mathbb{F}_{11}$, as follows
\begin{align}
\underline{x}_{f}(1) = \left(\Theta_{\mathcal{H}}(1)\Upsilon_{\mathcal{H}(1)}\right)^{\intercal}+5\left(\Xi_{\mathcal{H}}(1)\Upsilon_{\mathcal{H}}(1)\right)^{\intercal}. \nonumber
\end{align}

To start the second step, $i=2$, of the repair decoding, then the repair decoder first uses $e_{f}^{(i-1)m(k-1)} = 7^{2} = 5$, to calculate 
\begin{align}
S_{2} \left(e_{f}^{(i-1)m(k-1)}\underline{\phi}_{f}\right)^{\intercal} = 5\left(\Xi_{\mathcal{H}}(1)\Upsilon_{\mathcal{H}}(1)\right), \nonumber
\end{align}
where $5\left(\Xi_{\mathcal{H}}(1)\Upsilon_{\mathcal{H}}(1)\right)$ is already derived in the step 1. Then the decoder calculates $\hat{\Upsilon}_{\mathcal{H}}(2)$ using equation (\ref{eq_UpsilonHat_i}) as,
\begin{align}
\hat{\Upsilon}_{\mathcal{H}}(2) &=  \Upsilon_{\mathcal{H}}(2)-\left[\begin{array}{c}
e_{1}^{((2-1)m-1)(k-1)} \underline{\phi}_{1} \\
\vdots \\
e_{4}^{((2-1)m-1)(k-1)} \underline{\phi}_{4}
\end{array} \right] S_{2} \left(e_{f}^{(2-1)m(k-1)}\underline{\phi}_{f}\right)^{\intercal} \nonumber \\
&= \left[ \begin{array}{c}
r_{2}(1,f) \\
r_{2}(2,f) \\
r_{2}(3,f) \\
r_{2}(4,f)
\end{array} \right]-\left[\begin{array}{c c}
1~ & ~1 \\
1~ & ~2 \\
1~ & ~3 \\
1~ & ~4 
\end{array} \right] S_{2} \left(5\left(\Xi_{\mathcal{H}}(1)\Upsilon_{\mathcal{H}}(1)\right)\right). \nonumber
\end{align}

From (\ref{eq_Mi_repair}), with 
\begin{align}
\Omega_{\mathcal{H}}(2) = \left[ \begin{array}{c c c c}
1~ & ~1~ & ~1~ & ~1 \\
4~ & ~8~ & ~5~ & ~10 \\
9~ & ~5~ & ~4~ & ~1 \\
5~ & ~9~ & ~3~ & ~1
\end{array}\right], \nonumber
\end{align}
the repair decoder then has access to 
\begin{align}
\hat{\Upsilon}_{\mathcal{H}}(2)&=\Omega_{\mathcal{H}}(2) \left[\begin{array}{c}
 M_{2}  \\
 S_{4}
\end{array} \right] \left( \underline{\psi}_{f}(2)\right)^{\intercal} \nonumber \\
&=\Omega_{\mathcal{H}}(2)\left[\begin{array}{c}
\left[\begin{array}{c c}
s_{7}~ & ~s_{8} \\
s_{8}~ & ~s_{9}
\end{array}\right] \\
\left[\begin{array}{c c}
s_{10}~ & ~s_{11} \\
s_{11}~ & ~s_{12}
\end{array} \right]
\end{array} \right] \left[\begin{array}{c}
5 \\
2
\end{array} \right]. \nonumber
\end{align}

As expected $\Omega_{\mathcal{H}}(2)$ is also an invertible matrix in the code alphabet $\mathbb{F}_{11}$, and we have
\begin{align}
\left(\Omega_{\mathcal{H}}(2)\right)^{-1} = \left[\begin{array}{c}
\Theta_{\mathcal{H}}(2) \\
\Xi_{\mathcal{H}}(2)
\end{array}\right] = \left[\begin{array}{c}
\left[ \begin{array}{c c c c}
4~ & ~4~ & ~9~ & ~2 \\
3~ & ~1~ & ~9~ & ~0 
\end{array}\right] \\
\vspace{2mm}\left[\begin{array}{c c c c}
7~ & ~10~ & ~1~ & ~2 \\
9~ & ~7~ & ~3~ & ~7 
\end{array}\right]
\end{array} \right]. \nonumber
\end{align}

Finally, the decoder calculates the lost data $\underline{x}_{f}(2)$, using (\ref{eq_xf1_repair}), with $e_{f}^{2m(k-1)}=7^4=3$ in $\mathbb{F}_{11}$, as follows
\begin{align}
\underline{x}_{f}(2) = \left(\Theta_{\mathcal{H}}(2)\Upsilon_{\mathcal{H}(2)}\right)^{\intercal}+3\left(\Xi_{\mathcal{H}}(2)\Upsilon_{\mathcal{H}}(2)\right)^{\intercal}. \nonumber
\end{align}
\end{example}

\section{Discussion and properties}\label{Sec_Discussion}

In this section we will briefly review some of the technical requirements and properties of the coding scheme presented in this work. Particularly, we will show that the field size and subpacketization level requirements of the presented coding scheme are not significantly limiting factors in the practical implementations. 

\subsection{Field Size Requirement}

The only factor that influences the choice of the code alphabet $\mathbb{F}_{q}$ in the presented coding scheme is the existence coefficient matrix $\Psi$, and all the inverses of its square submatrices. To satisfy this requirement for a network with $n$ storage nodes, it is enough to have $q \geq n$ \cite{VandermondeInverse}, which is the same as the field size requirement of many other coding schemes such as the Product Matrix codes \cite{PM_Codes} or the commonly used Reed-Solomon codes \cite{RS_Codes}. Hence, the presented coding scheme in this work is providing the bandwidth adaptivity property at no extra cost in the field size requirements, and any field of size larger than $n$ could be used as the code alphabet. It worth mentioning that the field size requirement of the only other bandwidth adaptive exact repair MSR constructions, introduced in \cite{ExplicitMSR} is lower bounded by $n^2-kn$ which is significantly larger for large distributed storage networks $n$. Moreover, techniques such as those presented in \cite{Raviv_Any_Field} could easily be applied in the presented coding scheme to reduce the field size to any arbitrary (e.g. binary) small field. However, this discussion is out of the scope of this work.

\subsection{Subpacketization Level}

The requirement for the subpacketization level $\alpha$ for the presented coding scheme is given in by (\ref{eq_alpha_design}). A natural question that arises is how fast does $\alpha$ grow as a function of the code design parameters, $k$, and $\delta$. From (\ref{eq_alpha_design}) it is clear that $\alpha$ is proportional to the design parameter $k$. Regarding the dependency of $\alpha$ on $\delta$, from \cite{LCMn_Nair} we have, 
\begin{align}
\mathrm{lcm}\left(1,2,\cdots,\delta \right) \leq 4^{\delta}, \nonumber
\end{align}
 and, recently \cite{LCMn_book} showed that
\begin{align}
\mathrm{lcm}\left(1,2,\cdots,\delta \right) \geq 2^{\delta},~~\text{for}~\delta \geq 7.  \nonumber
\end{align}

As a result, for the presented coding scheme, we have
\begin{align}
(k-1) 2^{\delta} \leq \alpha \leq (k-1) 4^{\delta},~~\text{for}~\delta \geq 7.  \nonumber
\end{align}

\section{Conclusion}
We presented an alternative solution for exact repair MSR codes in which optimal exact repair is guaranteed simultaneously with a range of choices, $D=\{d_1,\cdots,d_\delta\}$, for the number of helpers. The introduced coding scheme is based on the Product Matrix framework, introduced first in \cite{PM_Codes}. The repair mechanism in this framework is based on specific symmetries in the structure of the \emph{message matrix}. We proposed a novel structure for the message matrix, which preserves the required symmetries in many different submatrices. Corresponding repair mechanisms are also introduced to utilise these symmetries to perform optimal repair, whit different choices for the number of helpers, namely $d_{i}=(i+1)(k-1),~i\in\{1,\cdots,\delta\}$. In addition, the data reconstruction procedure is enhanced based on a novel successive interference cancellation scheme to perform optimally under the new design of message matrix. Comparing to the only other explicit constructions with exact optimal bandwidth adaptive repair, presented in \cite{ExplicitMSR}, we showed that when $d_{i}\geq 2k-2,~\forall{d_{i} \in D}$, the required values for $\alpha$, is reduced from $z_{\delta}^{n}$ to $kz_{\delta}$ for a DSS with $n$ nodes, and $k$ systematic nodes, which also reduces $\beta$ exponentially.

\appendices
\section{Proof of Lemma \ref{Lem_PMreconstruction}}\label{APP_Reconstruction}

Multiplying both sides of (\ref{eq_reconstruction_step}) by $\Phi^{\intercal}$ from right we get
\begin{align}\label{eq_X_symmetric}
X\Phi^{\intercal} = \Phi A \Phi^{\intercal} + \Delta \Phi B \Phi^{\intercal}.
\end{align}
Following the notation in \cite{PM_Codes}, we introduce
\begin{align}
P = \Phi A \Phi^{\intercal}, ~~Q = \Phi B \Phi^{\intercal}. \nonumber
\end{align}
Then using (\ref{eq_X_symmetric}) and the above equations we have
\begin{align}
X \Phi^{\intercal} = P+\Delta Q.
\end{align}
Note that both $P$ and $Q$ are symmetric $k \times k$ matrices. Recall that $\Delta$ is a diagonal matrix, with non-zero and distinct diagonal elements. Hence, for any $1\leq i<j\leq k$, we now have both
\begin{align}\label{eq_PQij}
\left(X \Phi^{\intercal}\right)_{i,j} = P_{i,j}+\Delta_{i,i} Q_{i,j},
\end{align}
and
\begin{align}\label{eq_PQji}
\left(X \Phi^{\intercal}\right)_{j,i} &= P_{j,i}+\Delta_{j,j} Q_{j,i} \nonumber \\
&=P_{i,j}+\Delta_{j,j} Q_{i,j}.
\end{align}
Then using (\ref{eq_PQij}) and (\ref{eq_PQji}), we recover all the non-diagonal elements of $P$ and $Q$. Now, let $\hat{\underline{p}}_{i}$, and $\hat{\underline{q}}_{i}$ denote the $i^{\text{th}}$ row of matrices $P$ and $Q$ excluding their diagonal elements. Moreover, let $\hat{\Phi}_{i}$ denote the submatrix of $\Phi$ derived by removing the $i^{\text{th}}$ row, and finally let's denote the $i^{\text{th}}$ row of $\Phi$ by $\underline{\phi}_{i}$. Then for all $i\in\{1, \cdots, k\}$ we have 
\begin{align}
\hat{\underline{p}}_{i} = \underline{\phi}_{i} A \left(\hat{\Phi}_{i}\right)^{\intercal}, \nonumber \\
\hat{\underline{q}}_{i} = \underline{\phi}_{i} B \left(\hat{\Phi}_{i}\right)^{\intercal}, \nonumber
\end{align}
and using Lemma \ref{Lem_GVM}, $\left(\hat{\Phi}_{i}\right)^{\intercal}$ is an invertible matrix, as it is a transposed generalized Vandermonde matrix. Then we calculate the following two matrices using the above equations for $i\in\{1, \cdots, k-1\}$,
\begin{align}
\left[\begin{array}{c}
\underline{\phi}_{1} \\
\vdots \\
\underline{\phi}_{(k-1)}
\end{array}\right] A = \hat{\Phi}_{k} A,~\left[\begin{array}{c}
\underline{\phi}_{1} \\
\vdots \\
\underline{\phi}_{(k-1)}
\end{array}\right] B = \hat{\Phi}_{k} B. \nonumber
\end{align} 
Now since $\hat{\Phi}_{k}$ is invertible we have both $A$, and $B$.


\bibliographystyle{IEEEtran}
\bibliography{IEEEabrv,DSS_Bibliography}

\end{document}